%% file: bare_jrnl.tex
\documentclass[journal]{IEEEtran}

\usepackage{graphicx}
\usepackage{ifpdf,times,bbm,units,array,centernot,algorithmic,amssymb,amsthm,color}
\usepackage{amsmath,multirow,framed,algorithmic,enumitem,multicol,booktabs}
\usepackage[vlined,ruled,noend]{algorithm2e}
\usepackage{cancel,cite,verbatim,footnote,multirow,booktabs,mathtools,enumerate,url}
\usepackage{dblfloatfix}    

\usepackage{subcaption}

\newtheorem{definition}{Definition}
\newtheorem{lemma}{Lemma}
\newtheorem{theorem}{Theorem}
\newtheorem{corollary}{Corollary}
\newtheorem{fact}{Fact}
\newtheorem{assumption}{Assumption}

\begin{document}

\title{Performance Analysis of RSU-based Multihomed Multilane Vehicular Networks}

\author{Saadallah~Kassir$^*$, Pablo~Caballero$^*$, Gustavo~de~Veciana$^*$, Nannan~Wang$^\dagger$, Xi~Wang$^\dagger$, Paparao~Palacharla$^\dagger$ \\ $^*$The University of Texas at Austin, Electrical and Computer Engineering Department \\ $^\dagger$Fujitsu Laboratories of America, Richardson, TX}

\maketitle

\begin{abstract}
Motivated by the potentially high downlink traffic demands of commuters in future autonomous vehicles, we study a network architecture where vehicles use Vehicle-to-Vehicle (V2V) links to form relay network clusters, which in turn use Vehicle-to-Infrastructure (V2I) links to connect to one or more Road Side Units (RSUs). Such cluster-based multihoming offers improved performance, e.g., in coverage and per user shared rate, but depends on the penetration of V2V-V2I capable vehicles and possible blockage, by legacy vehicles, of line of sight based V2V links, such as those based on millimeter-wave and visible light technologies. This paper provides a performance analysis of a typical vehicle's connectivity and throughput on a highway in the free-flow regime, exploring its dependence on vehicle density, sensitivity to blockages, number of lanes and heterogeneity across lanes. The results show that, even with moderate vehicle densities and penetration of V2V-V2I capable vehicles, such architectures can achieve substantial improvements in connectivity and reduction in per-user rate variability as compared to V2I based networks. The typical vehicle's performance is also shown to improve considerably in the multilane highway setting as compared to a single lane road. This paper also sheds light on how the network performance is affected when vehicles can control their relative positions, by characterizing the connectivity-throughput tradeoff faced by the clusters of vehicles.   \end{abstract}

\begin{IEEEkeywords}
Vehicular Networks, Vehicle-to-Vehicle, Vehicle-to-Infrastructure, Clustering, Multihoming, Multilane model
\end{IEEEkeywords}

\IEEEpeerreviewmaketitle

\section{Introduction}
\label{sec:introduction}
\input{introduction.tex}

\section{V2V+V2I Network Model}\label{sec:model}
\input{model.tex}

\section{V2V Cluster Characterization}
\label{sec:cluster}
\input{clusters.tex}

\section{Performance analysis} \label{sec-perf_analysis} 
\label{sec:perf_analysis}
\input{perf_analysis.tex}
\section{Extension to Multilane Highways}
\label{sec:multilane}
\input{multilane_ext.tex}

\section{Multilane Performance Evaluation}
\label{sec:multilane_perf_eval}
\input{perf_eval.tex}
\section{Revisiting the Poisson Assumption}
\label{sec:poisson}
\input{poisson.tex}

\section{Conclusion}
\label{sec:conclusion}
\input{conclusion.tex}

\bibliographystyle{IEEEtran}
\bibliography{bare_jrnl}

\appendix
\input{appendix.tex}

\end{document}

%% file: introduction.tex

The automotive industry is undergoing several disruptive changes that are likely
to have a significant impact on future wireless networks. These include:
(1) the emergence and increased use of ride sharing fleets; 
(2) the expected development and adoption of driverless car technologies which
may require robust V2V and V2I connectivity; and (3) the large volumes of wireless traffic 
generated by commuters free to work/play while on the road.
This paper embraces these changes by focusing on 
opportunities to leverage V2V-V2I capable fleets to deliver improved
connectivity to vehicles and offload 
traffic from the traditional cellular infrastructure.  
In particular we consider a network architecture wherein V2V-V2I capable vehicles
form relay network clusters which in turn use V2I links
to connect to possibly several Road Side Units (RSUs), leveraging multihomed connectivity.
The central goal of this paper 
is to model and study the performance and tradeoffs afforded by  
such vehicular wireless network architectures and their ability 
to address the potentially substantial traffic demands placed by 
future intelligent transportation network and future commuters in driverless vehicles.   
 
{\bf {\em Related work.} }
There has recently been substantial interest in enabling V2V connectivity driven in part by
the desire to improve safety, collaborative sensing and driving \cite{ZZC15}. 
Current Dedicated Short-Range Communications (DSRC) standards for V2V relaying 
are mature \cite{Ken11,LiW07,BhK14}, but in general fall short
at high vehicle densities or in highly dynamic environments
\cite{CCC14,CVG16, CHS16}. 
DSRC also supports V2I connectivity but only to nearby Road Side Units (RSUs) whence
their placement is critical \cite{NiH10,LMN13,YMY10}.
New alternatives based on millimeter-wave (mmWave) and Visible Light Communication (VLC) physical 
layers  that can deliver higher capacity, e.g., 1-10 Gbps, 
are being currently explored \cite{CCC14, CVG16, CaD17}. 
While these provide substantial improvements in capacity, they 
typically require Line of Sight (LoS) based connectivity. 
The network architecture studied in this paper also provides a partial solution to overcome LoS blockages through the diversity provided by multihomed multilane V2V-based vehicle clusters, making the network more robust to V2V and V2I blocking. 

This paper targets a deeper performance study of a network architecture leveraging RSUs and V2V clustering.
In the past, several works have analyzed such networks, characterizing the user association expected delays \cite{AbZ11,RSN14}, throughput \cite{AKA17, CML17}, connectivity \cite{SLZ14, ZCG16, PaW17, KKS16}, re-healing connection time \cite{WBM07, SoT11} and percolation (full-connectivity) probability \cite{KKS16}, among others. 

Our work builds up on models and results presented in these studies, but includes several novel aspects that were not tackled in the mentioned papers.

\begin{figure*}[t!]
\centering
\includegraphics[width=0.97\textwidth]{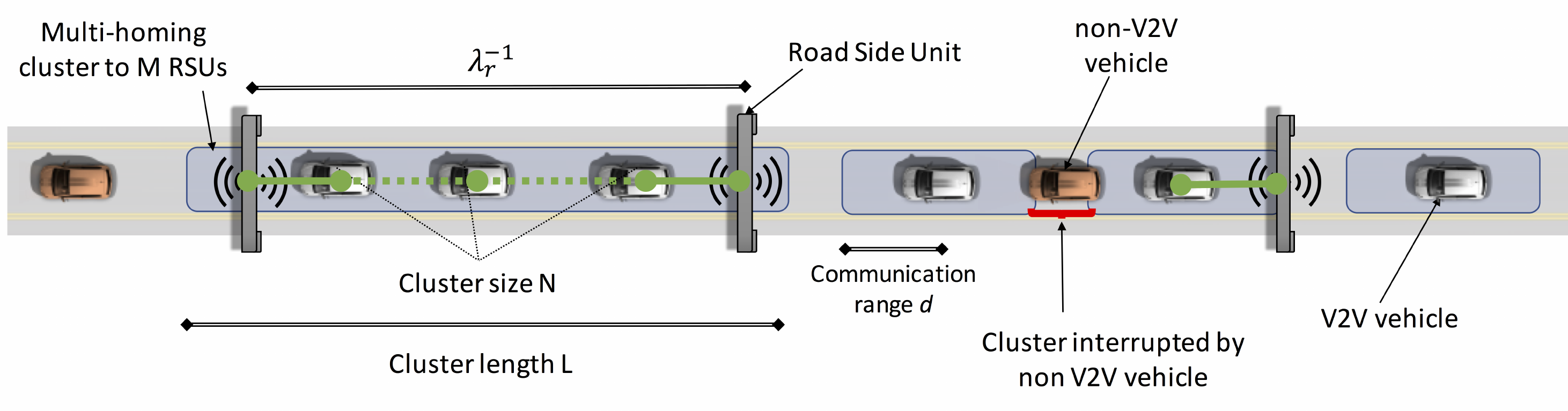}
\vspace{-0.2cm}
\caption{Example of the single lane highway modeled.}
\label{fig:example_highway}
\vspace{-0.2cm}
\end{figure*}

{\bf {\em Contributions.} }
This paper examines a model capturing the salient features of a vehicular-based wireless network, and expands previous work along several key directions. Our primary goal is to characterize the
ability of such networks to deliver high capacity data rates to vehicles reliably.  

First, we consider the role of V2V cluster RSU {\em multihoming}, i.e., the potential
benefits of enabling V2V clusters to connect to multiple RSUs at the same time, in terms of improved connectivity and reliability, as well as reduced variability in users' shared rate. 

Second, we provide an analytical framework to evaluate the network performance
which not only accounts for the role of multihoming, but also captures the impact
of V2V blockages and {\em market penetration} of V2V and V2I capable vehicles. We evaluate the sensitivity of a typical vehicle performance to market penetration.

Third, our evaluation of such vehicle-based networks suggests that even with a moderate penetration of V2V-V2I capable vehicles one can achieve improved connectivity and stability in  per-user shared rate. For instance, \textit{users see a reduced variability in their rate} and users in large stable clusters remain connected for large periods of time. 
Indeed, at high vehicle densities, one can expect almost deterministic user perceived performance. 
Comparisons with simple V2I networks which do not leverage V2V relaying are used to quantify the gains of the cluster-based architecture. 

Fourth, we propose a novel framework to study a typical vehicle's performance on multilane highway systems. This analysis provides key insights regarding the generalization of the single-lane results derived in the paper, as well as the performance of various traffic patterns such as vehicle intensity heterogeneity across the highway lanes, or lanes restricted to V2V-capable vehicles. 


Fifth, we validate the model's underlying assumptions and analyze the multilane highway performance based on additional system level simulation results of realistic traffic flows on roads, and revisit the assumptions to understand how idealized control of the vehicle distribution could lead to improved performance. The analysis of this best case scenario naturally leads to the introduction of a throughput-connectivity tradeoff.

Overall, these results show that such a network could provide a reliable means to offload substantial traffic from the cellular infrastructure to vehicles, particularly when the vehicle density is high, i.e., when such assistance is most needed.

{\bf {\em Paper organization.} }
The paper is organized as follows. 
Section \ref{sec:model} presents a single lane model for a V2V+V2I based wireless network architecture. 
Section \ref{sec:cluster} develops an analytical  characterization of the statistics of typical V2V clusters, e.g., the distributions of the number of vehicles, length and number of connected RSUs as a function of system parameters including the penetration of V2V-V2I capable vehicles. 
Section \ref{sec:perf_analysis}, provides a performance analysis of the coverage probability, shared rate 
and service redundancy as seen by a typical vehicle, and comparisons with those achieved by a V2I network. Section \ref{sec:multilane} provides an extension to multilane highways to assess the impact of non-homogeneities in lane traffic. 
We evaluate the performance of such systems in Section \ref{sec:multilane_perf_eval} and discuss the results sensitivity to the vehicle distribution assumption in Section~\ref{sec:poisson}. Finally, we present conclusions in Section \ref{sec:conclusion}.

%% file: model.tex
We first consider a model for an infinite straight \textit{single lane} road as in \cite{KKS16}.
The model corresponds to a snapshot of a collection of vehicles along the road, whose locations follow a Poisson Point Process (PPP) $\Phi_v$ with intensity $\lambda_v$ (vehicles/meter). The validity of the Poisson model has been discussed in empirical studies such as \cite{WBM07,GSH11} showing that such a model remains appropriate in settings under the so-called free-flow traffic conditions. We model the market penetration of V2V+V2I enabled vehicles on the road as a randomly chosen fraction $\gamma$ of vehicles. Thus a fraction $(1-\gamma)$ are legacy vehicles without communication capabilities which may block LoS communications among V2V capable vehicles. Furthermore, it follows that the locations of V2V capable vehicles follow a PPP with intensity $\gamma \lambda_v,$ and those of legacy vehicles a PPP with rate $(1-\gamma) \lambda_v$. We will use the term \textit{full market penetration} to denote $\gamma=1$.

Finally, RSUs are equally spaced each $\lambda_r^{-1}$ meters along the road. RSUs are wired to the Internet infrastructure to provide mapping data, infotainment and cloud computing services and may also relay messages to other clusters/vehicles. A depiction of the geometry of the network is displayed in Figure \ref{fig:example_highway}.

\textbf{\textit{Connectivity:}} We model the vehicle connectivity based on the three assumptions listed below. 

\begin{assumption}
We assume a unit disk connection model for V2V and V2I links.
\label{assum:1}
\end{assumption}

More specifically, a link is established if the destination vehicle is within a communication range of radius $d$ meters of the transmitter (vehicle or RSU), as in \cite{RSN14, MaA09}, and the LoS between their antennas is not obstructed, e.g., by another vehicle. We assume that the LoS between RSUs and cluster-head vehicles is never obstructed, e.g., by having RSUs above the road as illustrated in Figure~\ref{fig:example_highway}. 

\begin{assumption}
We assume that $d < \lambda_r^{-1}/2$. 
\label{assum:2}
\end{assumption}

Indeed, the communication range $d$ would typically be on the order of 10-200 meters, while RSUs might be deployed at a distance $\lambda_r^{-1}$ on the order of a few kilometers apart.

\begin{assumption}
We assume V2V links to have very high capacity, exceeding the maximum RSU capacity $\rho^{\text{RSU}}$. 
\label{assum:3}
\end{assumption}

Those links can be for instance based on mmWave or VLC technologies \cite{UGB15, CCC14, CVG16}, while the V2I links have maximum capacity of $\rho^{\text{RSU}}$. Thus, for simplicity, V2V links are not a bottleneck in this system.
One possible scenario that can be envisaged is using VLC technology, known for its considerable bandwidth \cite{CaD17} for LoS V2V links, while cellular (potentially mmWave based) links are used for V2I. Another scenario would be that both V2V and V2I links run on the same technology, but on orthogonal channels, and where the V2V channels can be designed to be of larger bandwidth than the V2I ones.
Some other multi-RAT network design considerations to avoid throughput bottlenecks are presented in \cite{KVB16}.






The above assumptions capture the salient features of V2V+V2I networks, allowing us to explore their fundamental characteristics of possible deployments.

In this work, we focus on analyzing the downlink performance of this network architecture.

\textbf{\textit{Sharing / Scheduling:}}
V2V capable vehicles within communication range can form V2V relaying clusters.
In this paper, we assume \textit{RSU multihoming}, i.e., a cluster can connect to multiple RSUs in its range, as illustrated in Figure \ref{fig:example_highway}. This enables the vehicles to see ($i$) improved performance, i.e., connectivity and reduced variability, by sharing the capacity of multiple RSUs and ($ii$) improved reliability through infrastructure redundancy in the case of link failures. In this work, for simplicity and given Assumption~\ref{assum:3}, we will use a max-min fair resource allocation among the vehicles and clusters; where a resource allocation is said to be \textit{max-min fair} if it is only possible to increase the resources assigned to a vehicle by decreasing the rates of vehicles which have lower rates \cite{BGH87}. We study the downlink shared rate seen by vehicles assuming the V2V-V2I capable vehicles are always active, i.e., full buffer traffic. 

\textbf{\textit{Benchmark system:}} 
We compare the described \textit{V2V+V2I} multihoming architecture with the same \textit{V2I} network but \textbf{without} V2V relaying, i.e., where vehicles \textbf{do not} relay data to form \textit{V2V} relaying clusters and are only be connected to the infrastructure if they are within range of an RSU.

%% file: clusters.tex
\begin{definition}(\textbf{Vehicle Relay cluster})\label{cluster_def} A V2V relay cluster is a group of vehicles that can inter-communicate without the network infrastructure, i.e., each vehicle has a direct connectivity link with at least one other vehicle in its cluster.
\end{definition}

A typical cluster is characterized by $(N, L, M),$ where $N$ and $L$ are random variables denoting the number of vehicles and communication range (length) of the cluster, respectively; and $M$ denotes the number of RSUs that the cluster is connected to, see Figure \ref{fig:example_highway}. 
The performance analysis will be based on characterizing the joint distribution of $N, L$ and $M.$

For readability purposes, the proofs of all theoretical results have been relegated to the Appendix.

\begin{lemma}(\textbf{Number of vehicles in a cluster})\label{lem:mdi}
The number of vehicles ${N}$ in a typical cluster follows a geometric distribution with parameter $\varphi=1-\gamma (1-e^{-\lambda_v d}),$ i.e.;
$$p_{N}(n)=\varphi\left(1-\varphi\right)^{n-1},$$
and $\mathbb{E}[N]={1}/{\varphi}$.
Consequently, under full market penetration, ${N}$ follows a geometric distribution with parameter $e^{-\lambda_v d}$.

\noindent(\textbf{Cluster communication length})
The typical cluster length ${L}$ distribution can be obtained by the inverse Laplace transform $\mathcal{L}^{-1}(\cdot)$ as follows:
$$f_{L}(l)=\mathcal{L}^{-1}\left(\frac{e^{-2sd} \varphi}{1-M_T(-s)+\varphi M_T(-s)}\right)(l),$$ 
where $$M_T(s)=\frac{\lambda_v e^{d (s-\lambda_v)}-\lambda_v}{(s-\lambda_v)(1-e^{-\lambda_v d})};$$
and the conditional distribution of $L$ given $N=n$ is given by
$f_{L\mid N}(l\mid n)=\mathcal{L}^{-1}\left(e^{-2 s d}\left[M_T(-s)\right]^{n-1}\right)(l).$
For the case of full market penetration, the cluster length ${L}$ distribution is:
{\color{black}$$f_{L}(l)=\mathcal{L}^{-1}\left(\frac{e^{-d (2s+\lambda_v)} \left(s+\lambda_v\right)}{s+\lambda_v e^{d (s-\lambda_v)}}\right)(l).$$ }

\noindent(\textbf{Number of connected RSUs})
The conditional CDF of the number of RSUs $M$ serving a cluster of length $L$  is given by
\begin{align*}
F_{M\mid L}(m\mid l)
&=\begin{cases}
1 & \mbox{if } m\lambda_r^{-1}<l\\
1-\frac{l}{m\cdot \lambda_r^{-1}}   & \mbox{if } (m-1) \lambda_r^{-1}< l \le m \lambda_r^{-1}\\
0 &  \mbox{otherwise} 
\end{cases}, 
\end{align*}
and the conditional CDF of the number of RSUs that serve a cluster with $N=n$ vehicles is:
\begin{align}
F_{M\mid N}\left(m\mid n\right)&=F^c_{L\mid N}((m-1)\lambda_r^{-1}\mid n) \nonumber\\&-\int\limits_{(m-1)\cdot \lambda_r^{-1}}^{m\cdot \lambda_r^{-1}} \frac{l\cdot f_{L\mid N}(l\mid n)}{m\cdot \lambda_r^{-1}} dl. \label{eq:M_given_N}
\end{align}
Finally, the CDF of $M$ is given by:
$$F_{M}\left(m\right)=\sum\limits_{n=1}^{\infty} p_{N}(n) F_{M\mid N}\left(m\mid n\right),\qquad {\color{black}\text{ for } m\in\mathbb{N}}.$$ 
\end{lemma}

\begin{figure*}[t!]
\includegraphics[width=\textwidth]{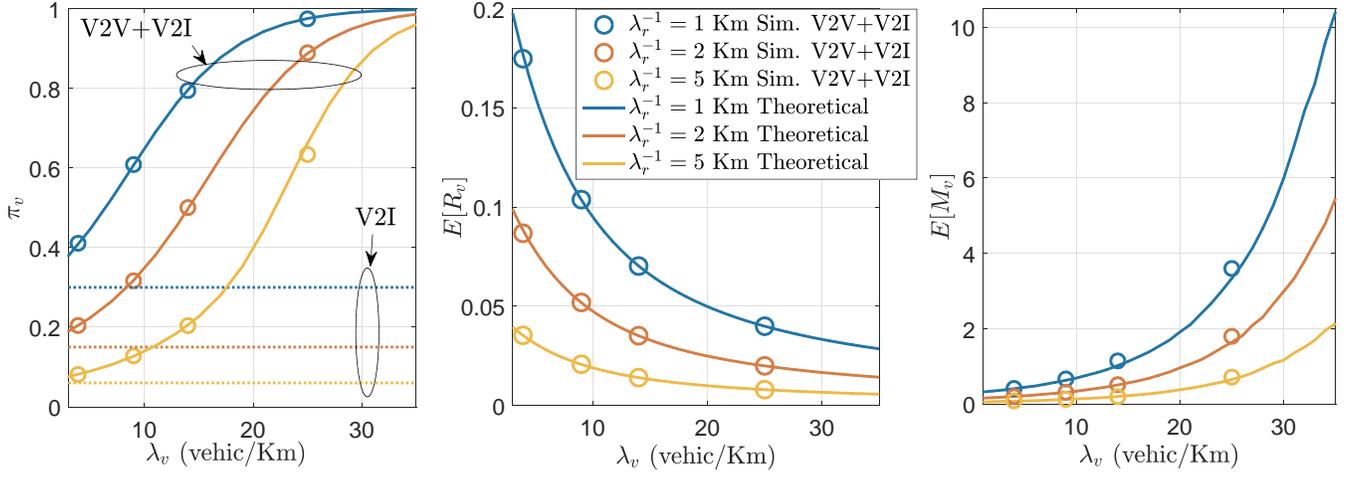}
\vspace{-0.9cm}
\caption{Leftmost subfigure: Vehicle coverage probability for V2V and no V2V cases.  Center: Expected RSU network throughput. Rightmost: Expected number of RSUs connected per typical vehicle. In all cases $d=150m$ and $\gamma=1$. Legend in the central figure applies to all plots.}
\label{fig:results1}
\end{figure*}

%% file: perf_analysis.tex
In this section, we analyze the performance of V2V+V2I multihoming networks and compare it to the V2I only network architecture. We distinguish performance metrics corresponding to the V2I networks via variables with an asterisk superscript, i.e., $R_v$ and $R_v^*$ will denote the rate of a typical vehicle in the V2V+V2I and the V2I only networks respectively. Also, we will evaluate the performance seen by a typical vehicle, indicating the related metrics by a subscript $v$.
\subsection{Typical Vehicle Coverage Probability}
In this section, we shall refer to the coverage probability as the probability that a typical vehicle is connected to one or more RSUs, either directly or through V2V relaying.
Clearly the benefit of the V2V+V2I network is that it allows vehicles to relay messages from RSUs, increasing the coverage probability. We let $\pi_v$ denote the probability that a typical vehicle is connected (possible through relaying) to the infrastructure.
Specifically, note that the typical vehicle coverage probability for the benchmark V2I network  is independent of the traffic intensity. By contrast, in the V2V+V2I network, higher traffic intensities lead to longer and bigger clusters, increasing the typical vehicle coverage probability. The following result addresses the coverage probability for both networks assuming $2d \leq \lambda_r^{-1}$.
\begin{lemma}(\textbf{Coverage probability})\label{lem:coverage} 
The coverage probability of a \textbf{typical vehicle} in the V2V+V2I network is given by:
\begin{equation}
\pi_v=\varphi^2 \cdot\sum\limits_{n=1}^{\infty} n \cdot(1-\varphi)^{n-1}\cdot {\color{black}F^c_{M\mid N}\left(0\mid n\right),}
\label{eq:coverage_V2V}
\end{equation}
where $F^c_{M\mid N}\left(0\mid n\right)$ is the probability that a cluster is connected to at least 1 RSU given $N=n$, given in Lemma~\ref{lem:mdi}.

The coverage probability of a typical vehicle in a V2I network is independent of $\lambda_v$ and given by:
\begin{equation}
\pi_v^{*}=\frac{2d}{\lambda_r^{-1}},\quad \mbox{ for } ~ ~ ~ 2d\le \lambda_r^{-1}.
\label{eq:coverage_noV2V}
\end{equation}
\end{lemma}

Numerical evaluations of \eqref{eq:coverage_V2V} and \eqref{eq:coverage_noV2V} are displayed in Figure \ref{fig:results1} (left). As expected, the coverage probability is always greater for V2V+V2I and increases rapidly to 1 with the traffic load intensity $\lambda_v$ on the road. Figure~\ref{fig:results_gamma} exhibits the coverage probability for V2V+V2I as a function of the penetration $\gamma$; it shows that the sensitivity of the coverage to the traffic intensity is higher at higher $\gamma$, e.g., for $\gamma=0.9$ where the coverage probability attains a maximum for $\lambda_v\approx 25$ vehicles/km and varies notably with  $\lambda_v$. Indeed, increasing $\lambda_v$ increases the effect of the blocking vehicles, reaching a regime where long clusters are not possible and where $\pi_v$ is independent of $\lambda_v$, consistently with \eqref{eq:coverage_noV2V}. Therefore, if $\gamma <1$, $\pi_v$ eventually decreases and converges back to the value presented in \eqref{eq:coverage_noV2V}.
\vspace{-0.2cm}

\begin{figure}[h]
        \centering        \includegraphics[width=\columnwidth]{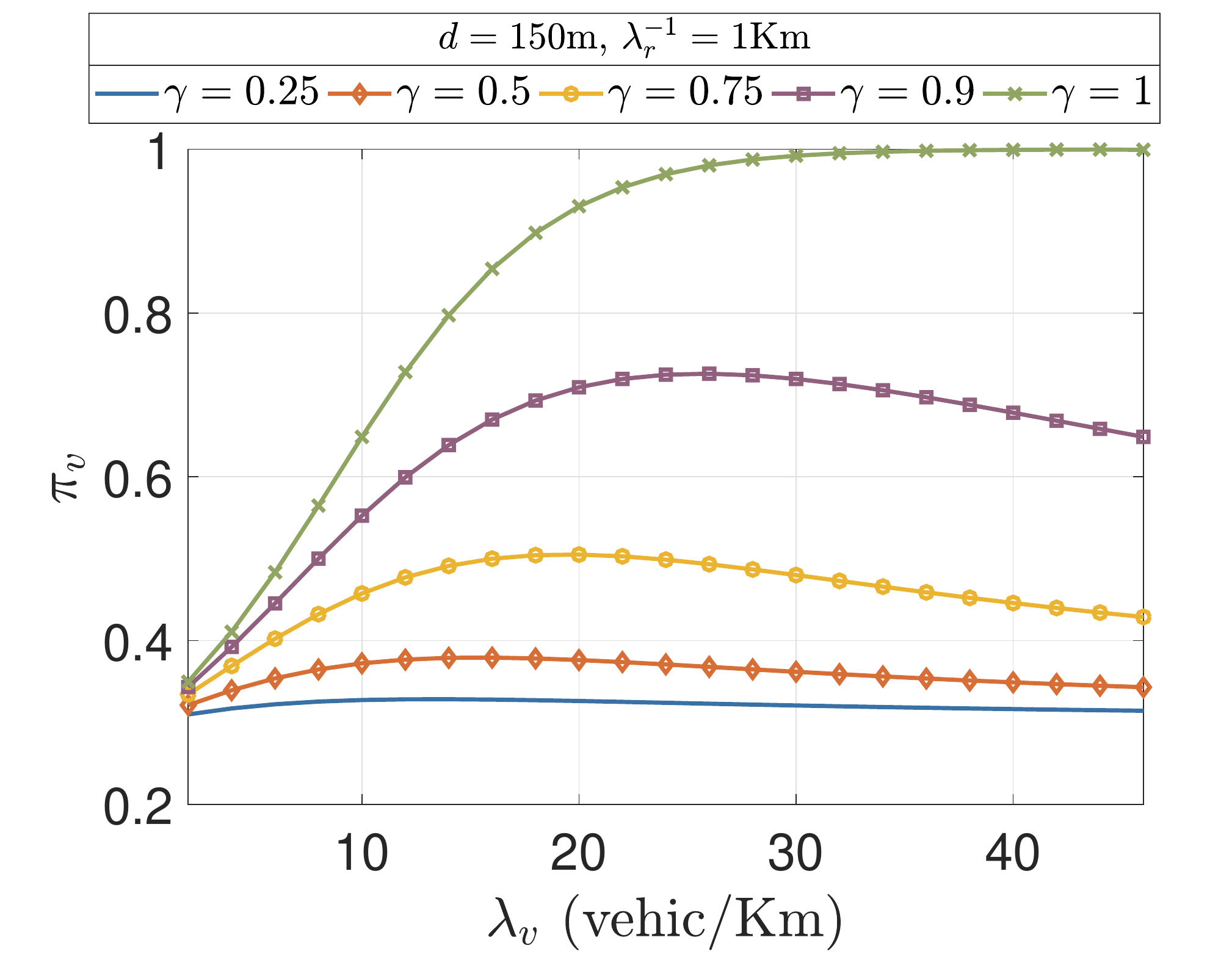}
        \vspace{-0.2cm}
        \caption{Impact of the load in the coverage probability for different market penetrations $\gamma$.}
        \label{fig:results_gamma}
\end{figure}
\vspace{-0.3cm}

\begin{figure}[h]
\centering
\includegraphics[width=\columnwidth]{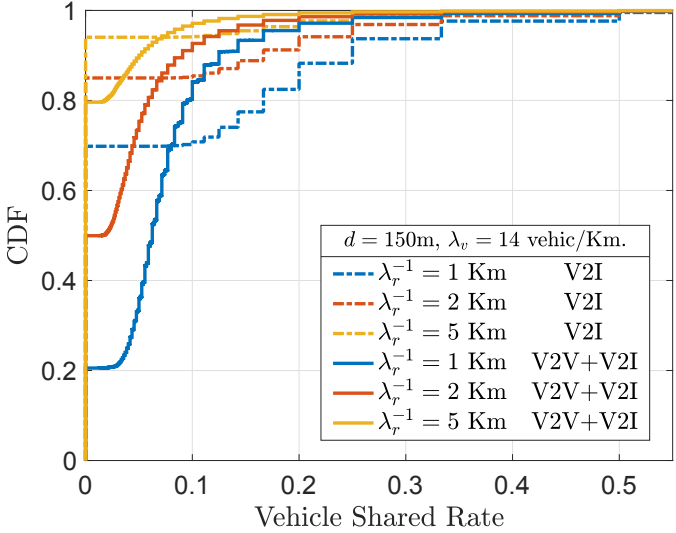}
\vspace{0.03cm}
\caption{Empirical CDF of the typical shared rate for V2I vs V2V+V2I and different inter-RSU distances ( $\gamma=\rho=1$).}
\label{fig:results_ecdf}
\end{figure}

\subsection{Typical Vehicle Shared Rate}
The shared rate seen by a typical vehicle is defined as its allocations of the multihomed RSU capacity of its cluster under max-min fair sharing and denoted by the random variable $R_v$. The shared rate, for both networks, i.e., V2V+V2I and V2I, thus depends on
$\lambda_v,\gamma, d, \rho^{\text{RSU}}$ and $\lambda_r^{-1}.$
\begin{theorem}(\textbf{Expected shared rate})~\label{thm:throughput}
The mean shared rate of a typical vehicle in the V2V+V2I and the V2I networks are equal, i.e.,  $\mathbb{E}[R_v]=\mathbb{E}[R_v^*]$ and given by:
\begin{equation}
\mathbb{E}[R_v] =\frac{\rho^{RSU}}{\gamma\lambda_v \lambda_r^{-1}}\left(1-e^{-2\gamma\lambda_v d}\right)\le\rho^{RSU} \frac{\mathbb{E}[M]}{\mathbb{E}[N]},
\label{eq:throughput}
\end{equation}
\end{theorem}

\noindent where $\mathbb{E}[M]$ and $\mathbb{E}[N]$ can be computed using Lemma~\ref{lem:mdi}.

{\color{black} Note that the mean rate for both architectures are equal because the number of busy RSUs is the same, independently of the underlying V2V connectivity. Assuming all vehicles are infinitely backlogged the overall downlink rate is the same and thus so is the mean rate per vehicle.}

Although V2V relaying collaboration does not alter the \textit{mean} shared rate seen by vehicles, see Figure \ref{fig:results1} (center); it significantly impacts the coverage probability and the shared rate \textit{distribution}.

\begin{theorem}(\textbf{Shared rate distribution}) \label{thm:cdf_rate}
The CDF of the shared rate in a V2V+V2I network  $R_v$ satisfies: 
\begin{align}
F_{R_v}(r)\ge 1- \varphi^2\sum\limits_{n=1}^{\infty}  n ~ (1-\varphi)^{n-1} F^c_{M\mid N}&\left(\Bigl\lceil\frac{r n}{\rho^{\text{RSU}}}\Bigr\rceil\mid n\right), \label{eq:cdfR}
\end{align}
and $P(R_v=0)=1-\pi_v$ while that in the V2I network is given by 
\begin{equation}
F_{R_v^{*}}(r)=1-\frac{2d}{\lambda_r^{-1}}\cdot Q\left( \frac{\rho^{RSU}}{r}-1 , ~ ~ 2~\gamma~\lambda_v~d \right),
\label{eq:cdfR_noV2V}
\end{equation}
where $Q$ is the regularized gamma function and $P(R_v^{*}=0)=1-\pi_v^{*}$.
Furthermore, ${R_v^{*}}\ge^{\text{icx}}{R_v}$, where \textit{icx} dominance\footnote{The definition for icx dominance is found in Definition \ref{def:dominance} in the appendix} implies:
\begin{equation}
\mbox{Var}({R_v^{*}})\ge\mbox{Var}({R_v}).
\label{eq:variance_rate}
\end{equation}
\end{theorem}

Numerical evaluations of \eqref{eq:cdfR} and \eqref{eq:cdfR_noV2V} are shown in Figure~\ref{fig:results_ecdf} and the resulting variability in Figure~\ref{fig:results_var}. These demonstrate the superiority of the V2V+V2I network architecture in terms of providing, not only improved connectivity, but also a substantial decrease in the shared rate variability of a typical user.
Note that in Figure~\ref{fig:results_var} we have plotted the dispersion of the per-user shared rate, defined as $\sigma/\mu,$ i.e., the standard-deviation over the mean of the per user shared rate. {\color{black}In addition, we have displayed the lower bound on the dispersion for the non-V2V scenario, given by the dispersion as $\lambda_v\to\infty$.}
{\color{black}It can be observed} that the rate dispersion converges to 0 for the V2V+V2I network. By contrast, in the V2I network the dispersion of the shared rate is bounded below. These results show that the V2V+V2I network at reasonably high vehicle density will provide them with an increasingly stable and almost deterministic shared rate to vehicles. 

\subsection{Multihoming Redundancy}
RSU multihoming provides connection redundancy to a cluster. This redundancy in principle improves the reliability of vehicle connectivity in presence of unreliable/obstructed V2I links. The following result follows immediately from \eqref{eq:throughput} in Theorem~\ref{thm:throughput}.
\begin{corollary} (\textbf{Multihoming / redundancy}) 
The \textbf{expected number of RSUs} $\mathbb{E}[M]$ per cluster  is bounded by:
\begin{equation}
\mathbb{E}[M]\ge \frac{ 1- e^{-2 \gamma \lambda_v d}}{\gamma\lambda_v \lambda_r^{-1} (1-\gamma+\gamma \cdot e^{-\gamma \lambda_v d})},
\end{equation}
which for full market penetration corresponds to
\begin{equation}
\mathbb{E}[M]\ge \frac{ e^{\lambda_v d}- e^{-\lambda_v d}}{\lambda_v \lambda_r^{-1}}=\frac{2\sinh(\lambda_v d)}{\lambda_v \lambda_r^{-1}}.
\end{equation}
\end{corollary}

As can be observed from this equation, $\mathbb{E}[M]$ i.e., the expected number of RSUs that the cluster of a typical vehicle is connected to grows rapidly with the traffic intensity $\lambda_v$ and the vehicle communication range $d.$ 
A similar trend is observed in Figure \ref{fig:results1} (right) where we have plotted $\mathbb{E}[M_v]$, the mean number of RSUs a typical vehicle would see its cluster connected to. We see a rapid increase in the expected number of RSUs as $\lambda_v$ increases. These results confirm an exponential growth of redundancy suggesting possibly substantial improvements in reliability of multihomed systems.
\begin{figure}[t]
\centering 
\includegraphics[width=1\columnwidth]{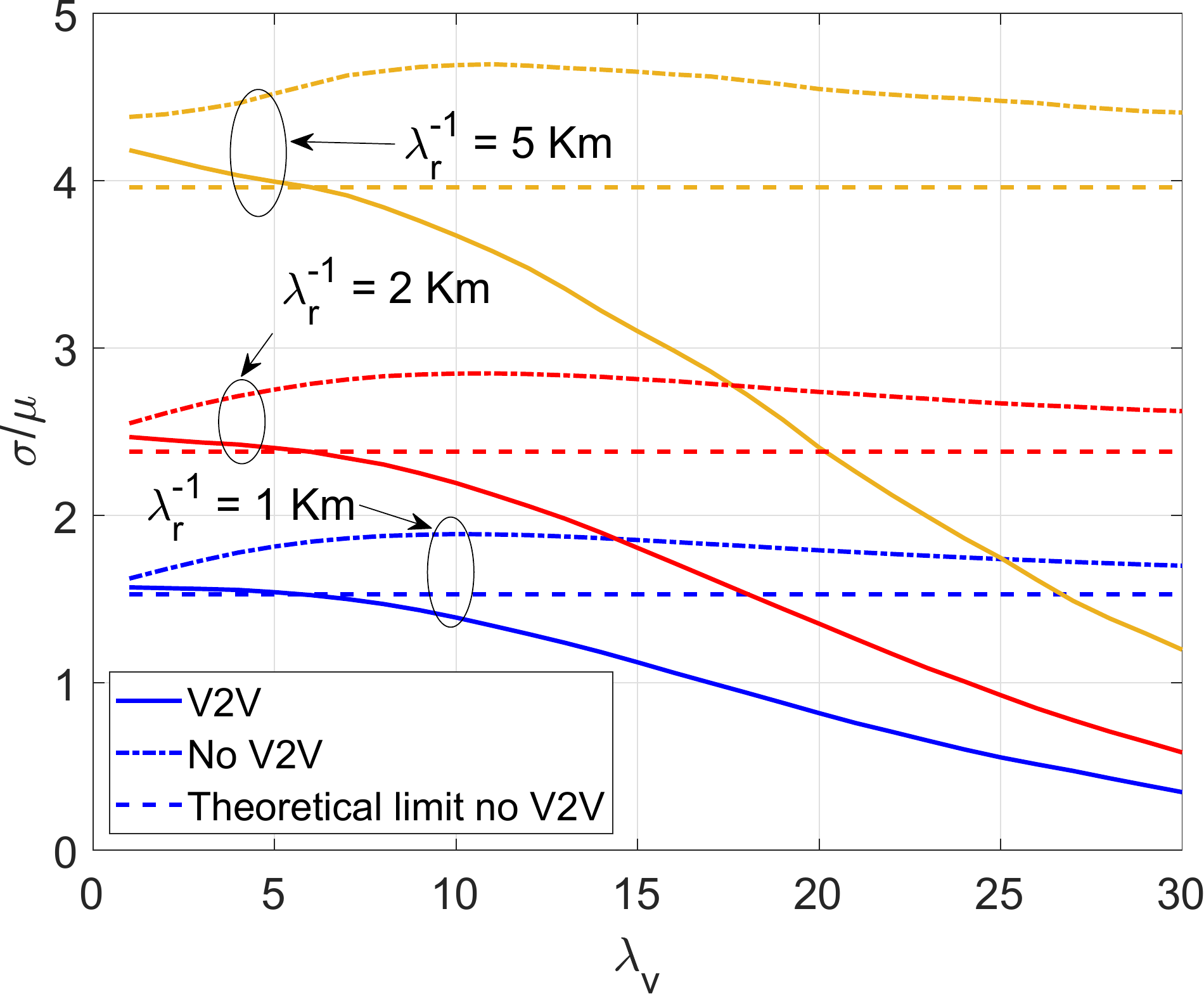} 
\caption{Dispersion (standard deviation over the mean) of the vehicle shared rates under {\color{black}V2V+V2I and V2I} only scenarios, and different inter-RSU distances ($d=150$m, $\gamma=1$).}
\label{fig:results_var}
\vspace{-0.15cm}
\end{figure}

The benefit of the redundancy is also reflected in Figure \ref{fig:results_ge2} which exhibits the probability that a typical vehicle benefits from multihoming as the vehicle intensity increases. This probability reaches values very close to 1 under heavy and congested traffic conditions, for the given values of $\lambda_r^{-1}$, providing evidence of the potential for higher reliability through multihoming.

\begin{figure}[h]
        \centering
        \includegraphics[width=0.95\columnwidth]{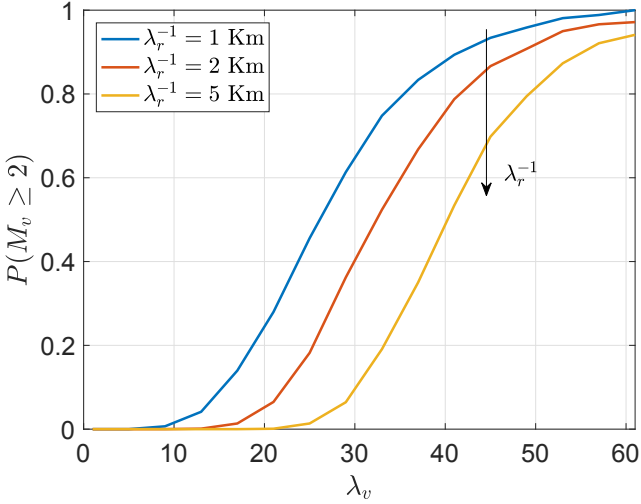}
        \caption{Redundancy: Probability for a typical vehicle cluster to be connected to 2 or more RSUs.}
        \label{fig:results_ge2}
\end{figure}

%% file: multilane_ext.tex
\begin{figure*}[t!]
\centering
\includegraphics[width=0.65\textwidth]{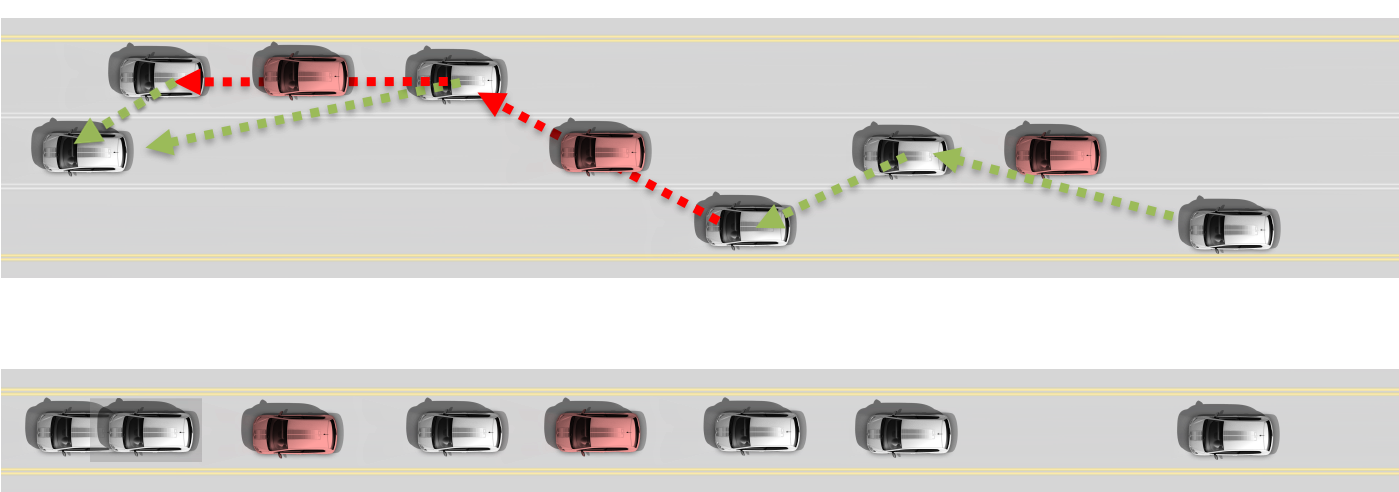}
\caption{Example of the multi-lane highway approximation construction. The bottom system is the construction proposed based on the rules in Definition \ref{block-model}.}
\label{fig:example_Multihighway}
\end{figure*}

The system described  in Section \ref{sec:model} and analyzed in Section \ref{sec-perf_analysis}  considers a single lane highway. In this section we consider multilane highways. Because an exact analysis is somewhat intricate we shall explore how one can relate the performance of multilane highways to the single lane setting.

\begin{definition}(\textbf{Multilane highway})\label{def:highway}
We define a multilane highway as a triplet: $\left(\eta, \boldsymbol{{\lambda^{\text{V2V}}}}, \boldsymbol{{\lambda}}^{b}\right),$
where $\eta$ is the number of lanes, which are indexed sequentially $1,2,\ldots,\eta$ and $$\boldsymbol{{\lambda^{\text{V2V}}}}\triangleq(\lambda^{\text{V2V}}_k: k=1,2,\ldots,\eta), \qquad\lambda^{\text{V2V}}\triangleq\sum_{k=1}^{\eta} \lambda^{\text{V2V}}_k$$ and $$\boldsymbol{{\lambda}}^{b}\triangleq(\lambda^{b}_k: k=1,2,\ldots,\eta),\qquad\lambda^{b}\triangleq\sum_{k=1}^{\eta} \lambda^{b}_k$$ correspond to the intensities of V2V capable and blocking legacy vehicles in each lane, respectively. We assume each lane has independent PPPs of vehicles, and distances among lanes are negligible as compared to the communication range~$d$. 
\end{definition}

\begin{definition}(\textbf{Multilane blocking model})\label{block-model}
In our multilane highway, LoS blocking is modeled as follows. Consider a triplet $(\mathrm{k}^-,\mathrm{k}^b,\mathrm{k}^+)$ as the lane index of the transmitter, of a potential blocker and the receiver, respectively. A blocker \textbf{may} obstruct the LoS link from $\mathrm{k}^{-}$ to $\mathrm{k}^{+}$ if and only if it is located in a lane between the transmitter and receiver, i.e.,
$$\mathrm{k}^-=\mathrm{k}^b=\mathrm{k}^+\quad\text{or}\quad \mathrm{k}^-<\mathrm{k}^b<\mathrm{k}^+\quad\text{or}\quad \mathrm{k}^->\mathrm{k}^b>\mathrm{k}^+.$$
From this definition, it follows that the worst case number of lanes where vehicles might be located and \textbf{might} block a LoS link is $k^*=\max(1,\eta-2).$ 
\end{definition}

For a typical vehicle in a multilane highway $\mathcal{M}$, we define the number of vehicles, length and number of multihomed RSUs to its cluster as $(N_v^\mathcal{M}, L_v^\mathcal{M}, M_v^\mathcal{M})$ for the multi-lane highway and $(N_v^\mathcal{S}, L_v^\mathcal{S}, M_v^\mathcal{S})$ for a single lane road $\mathcal{S}$. We will also define  $(\pi_v^\mathcal{M}, R_v^\mathcal{M})$ and $(\pi_v^\mathcal{S}, R_v^\mathcal{S})$ as the coverage probability and shared rate of a typical vehicle in multi and single lane highways.

\begin{theorem}\label{thm:multilane2}
For a given multilane highway 
$\mathcal{M}=(\eta, \boldsymbol{{\lambda^{\text{V2V}}}}, ~\boldsymbol{{\lambda^{b}}})$
let $\mathcal{S}=(1, ~\gamma\lambda,\lambda^b_{\text{eff}})$
be an associated single lane highway system
where:
$${{\lambda}}=\lambda^{\text{V2V}}+\lambda^{b}; \ \  \gamma=\frac{\lambda^{\text{V2V}}}{{\lambda}}  \ \ 
\text{ and } \ \  \lambda^b_{\text{eff}} = \max(\lambda^{b}_0, \lambda^{b}_k, \sum^{\eta-1}_{i=2}\lambda^{b}_{i}).$$
Then, it follows that \footnote{The definitions of $\le^{st}$ and $\le^{icx}$ dominance can be found in the appendix.}:
$$N_v^\mathcal{M}\ge^{st}N_v^\mathcal{S},\quad L_v^\mathcal{M}\ge^{st}L_v^\mathcal{S}\quad \text{and}\quad M_v^\mathcal{M}\ge^{st}M_v^\mathcal{S}$$
and
$$\pi_{v}^\mathcal{M}\ge \pi_{v}^\mathcal{S},\quad R_v^\mathcal{M}\le^{icx}R_v^\mathcal{S}.$$
In other words, the multilane highway has larger cluster statistics, better coverage and decreased variability relative to the associated single lane highway.
\end{theorem}
A high level illustration of our approach is depicted in Figure \ref{fig:example_Multihighway} and a sketch of the proof is provided in the appendix. The single lane performance can in turn be obtained by using the result in the previous sections.

%% file: perf_eval.tex
In this section, we further assess the performance of the proposed V2V+V2I network architecture via simulations. This will enable us to infer useful design and deployment strategies for the proposed collaborative technology in future vehicles and highways.
The communication range $d$ is set to be 150 meters and the inter-RSU distance $\lambda_r^{-1}$ is fixed at 1 Km, unless otherwise specified. Table~\ref{tab:typicalParameters} shows typical values for different parameters that were used in the simulations.  For our figures, we
have obtained $95\%$ confidence intervals achieving relative errors below
$2\%$ (not displayed).
In order to capture the effect of the blocking vehicles in the multilane system, we modeled vehicles as having a length of 5 meters allowing overlapping of vehicles resulting from the Poisson assumption on their location distribution.

\begin{table}[h]
\begin{tabular}{@{}ccccccc@{}}\toprule
\phantom{a}& \multicolumn{2}{c}{$d$ } & \phantom{ade}& \multicolumn{2}{c}{$\lambda_v$ }& \phantom{a} \\
\cmidrule{2-3} \cmidrule{5-6} 
& mmWave &VLC&  &Free-flow &Congestion&\\ 
& $75-200$m &$\approx100$m & &$\le 25$ veh./km  &$\ge 60$ veh./km&\\
\bottomrule\\\vspace{-0.8cm}
\end{tabular}
\caption{Typical parameter ranges in \cite{WBM07,NLJ15,PFH15}}
\vspace{-0.6cm}
\label{tab:typicalParameters}
\end{table}

\subsection{Homogeneous Multilane Highways}

\begin{figure}[h]
\centering
\includegraphics[width=\columnwidth]{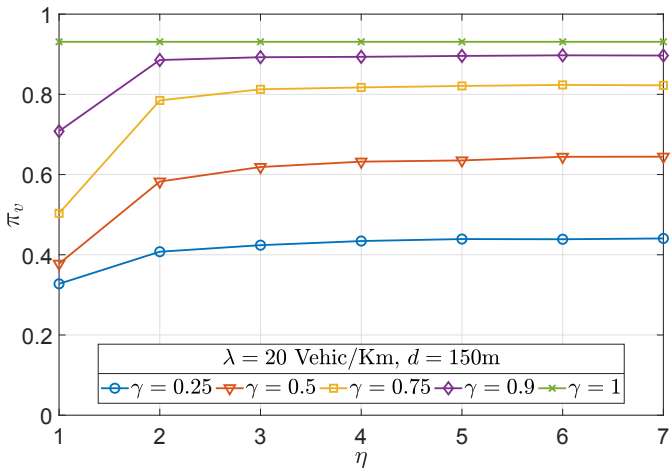}
\caption{Typical vehicle's coverage probability analysis as the driving ``degrees of freedom", i.e. $\eta$ increases, for different penetration rates.}
\label{fig:DoF}
\vspace{-0.4cm}
\end{figure}

\begin{figure*}[h]
        \centering
        \begin{subfigure}[h]{0.5\textwidth}
        \centering
  			  {{\includegraphics[width=0.97\textwidth]{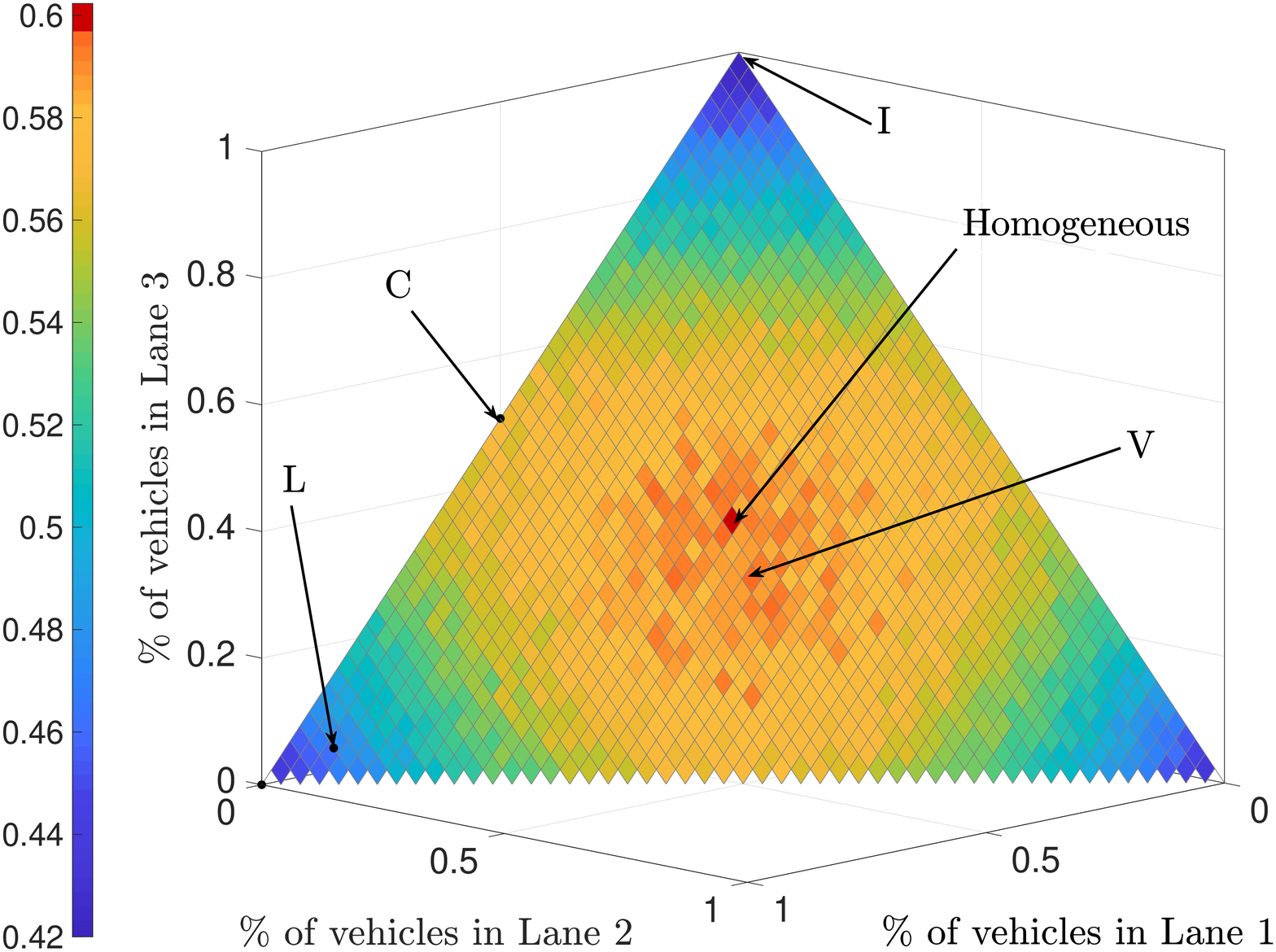}~\\~\\ }}%
                \caption{Coverage probability for $\eta=3$ lanes.}
        \end{subfigure}%
        \hfill
        \begin{subfigure}[h]{0.45\textwidth}
        \centering
   			   {{\includegraphics[width=1\textwidth]{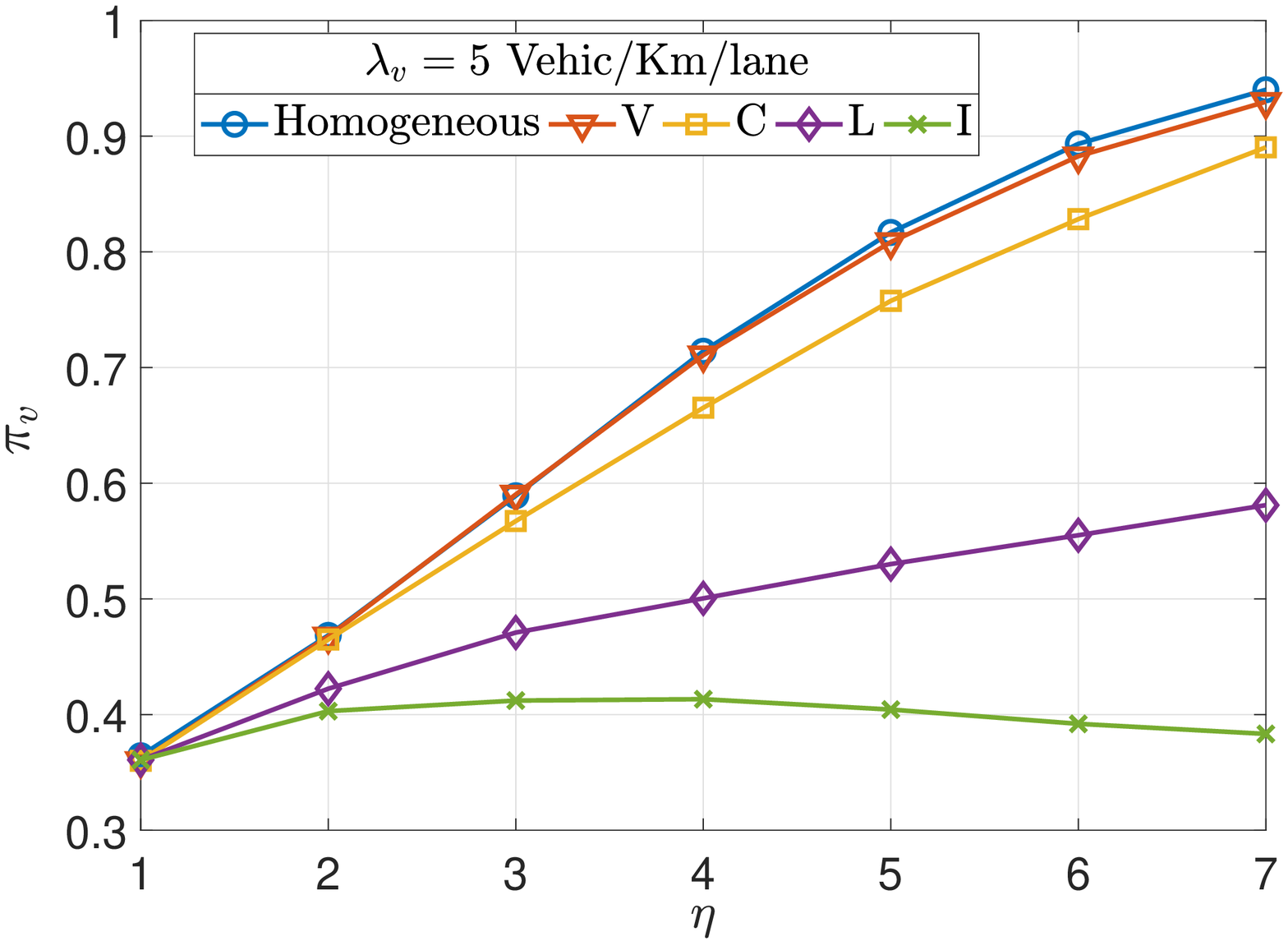}~\\~\\  }}%
                \caption{Coverage probability for different configurations.}
        \end{subfigure}%
    \caption{Multilane configuration coverage probability analysis for $\gamma=0.8, d=150$m,  $\lambda_r^{-1}=1$Km.}%
    \label{fig:configurations}%
    \vspace{-0.5cm}
\end{figure*}

Figure~\ref{fig:DoF} illustrates the variation in the coverage probability $\pi_v$ as $\eta$  increases, but the overall traffic intensity on the highway ($\lambda_v=20$ vehicles/Km) remains unaltered. This can be interpreted as the effect of increasing the vehicles' ``degrees of freedom"  to overcome  blocking by legacy vehicles.

A first observation is that the marginal gain in performance is most considerable when increasing the number of lanes from 1 to 2, while further increments in the number of lanes result in smaller relative gains. An explanation of this effect is that vehicles in the V2V+V2I network will see on average twice fewer blockers when passing from $\eta=1$ to $2$; while the relative decrease in the average number of blockers is smaller for higher values of $\eta$. Note that increasing the ``degrees of freedom" does not affect the performance of the system under full-market penetration as the same clusters will be formed for any value of $\eta$. From this result, one can infer that, as long as it is greater or equal than 2, the number of lanes of a highway,  will not substantially affect the connectivity probability.

\subsection{Heterogeneous Multilane Highways}
Next, we further explore the impact of heterogeneous traffic intensity across lanes on the coverage probability $\pi_v$. Note that such heterogeneity is typical in highways nowadays in a free-flow regime, since for instance a greater density of slower vehicles is seen in the right hand lanes. 
Figure~\ref{fig:configurations}(a) exhibits the effect of the vehicle distribution on a three-lane highway. 
In this figure, each coordinate represents the proportion of vehicles driving on each lane, therefore all possible configurations lie on the simplex. We observe that the homogeneous configuration has the best performance as it offers the best balance between minimizing the effect of blockers on the same and across lanes. The results show that performance deteriorates slowly when moving away from the homogeneous configuration, only experiencing notable decreases when moving to extreme distributions, e.g., all users are concentrated on one lane. In order to extrapolate these results to greater values of $\eta$ we define five different types of heterogeneous lane intensity distributions:
\begin{itemize}
    \item Homogeneous: all lanes have equal vehicle intensities, e.g. for $\eta=5$, $\boldsymbol{\lambda} = \lambda_v \eta \cdot [\frac{1}{5}, \frac{1}{5}, \frac{1}{5}, \frac{1}{5}, \frac{1}{5}]$.
    \item V: traffic is symmetrically and gradually concentrated around the leftmost and rightmost lanes of the highway, such that the intensity is minimized in the middle and maximized in the first and last lanes, e.g. for $\eta=5$, $\boldsymbol{\lambda} = \lambda_v \eta \cdot [\frac{1}{3}, \frac{2}{15}, \frac{1}{15}, \frac{2}{15}, \frac{1}{3}]$.
    \item C: traffic is restricted to two lanes with identical intensities while $\eta-2$ lanes are empty, e.g. for $\eta=5$, $\boldsymbol{\lambda} = \lambda_v \eta \cdot [\frac{1}{2}, 0, 0, 0, \frac{1}{2}]$.
    \item I: traffic is restricted to one lane with $\eta-1$ lanes empty, e.g. for $\eta=5$, $\boldsymbol{\lambda} = \lambda_v \eta \cdot [1, 0, 0, 0, 0]$.
    \item L: $90\%$ of traffic is in the first lane while the other $10\%$ is evenly distributed across the $\eta-1$ remaining lanes, e.g. for $\eta=5$, $\boldsymbol{\lambda} = \lambda_v \eta \cdot [\frac{9}{10}, \frac{1}{40}, \frac{1}{40}, \frac{1}{40}, \frac{1}{40}]$.
\end{itemize}
Figure~\ref{fig:configurations}(b) confirms the trends exhibited in Figure~\ref{fig:configurations}(a) as the number of lanes of the highway  increases. The homogeneous distribution remains best as compared to the V, C, L and I configurations.

We note that unlike in Figure~\ref{fig:DoF}, the total number of vehicles increases with $\eta$ in the highway system.

An interesting insight which can be inferred from these results is the idea that congested highways (large $\lambda_v$) may have a better connectivity performance than free-flowing systems, as the intensity distribution is typically uniform across all the lanes in such cases.

\subsection{V2V Segregation Impact}
While manufacturers progressively release new vehicle models equipped with the V2V+V2I technology, we envision a transition period during which the roads will be shared among the new V2V-enabled and older legacy vehicles. 
In order to accelerate the integration and the spread of new automotive technologies, policies restricting specific lanes to driverless and V2V-enabled vehicles only might be put into place. This is akin to the current concept of high-occupancy vehicle lane. 
We analyze the effect on the coverage probability of reserving the first lane for V2V-enabled vehicles and we will define $\alpha$ as the percentage of V2V-enabled vehicles driving on this lane, i.e. the first lane has a vehicle intensity of $\alpha\gamma\lambda_v$ with only V2V-enabled vehicles while the others are mixed and uniformly distributed. Figure~\ref{fig:segregation} shows the effect of $\alpha$ on the network performance.
\begin{figure}[!t]
\vspace{0cm}
\centering
\includegraphics[width=0.95\columnwidth]{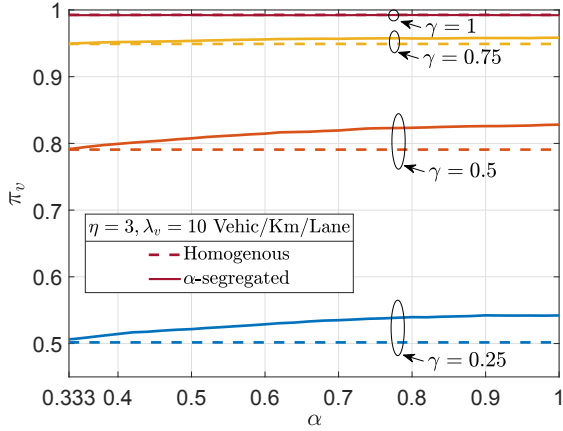}
\caption{Connectivity of $\alpha$-segregated scenario for different penetration rates.} \label{fig:segregation}
\vspace{-0.3cm}
\end{figure}
We observe that for $\alpha$ large enough, segregation does indeed improve coverage, particularly at low market penetration levels, implying that such a policy would lead to improved connectivity in the early stages of the V2V capable vehicles deployment. 

%% file: poisson.tex
In this section, we revisit one of the main assumptions of our network model, namely the Poisson distribution for cars on the highway. We study the validity of this assumption through realistic highway system simulations, before discussing the impact of different configurations on the network performance. Recall that as discussed in Section~\ref{sec:model}, this assumption was validated in part for the free-flow setting in \cite{WBM07,GSH11}.

\subsection{Validity of the Poisson Assumption}
We first explored the degree to which the PPP assumption might hold for a detailed simulation of vehicles on the road. The system-level simulator used is an enhanced version of the open-source automotive Intersection Management (AIM4) simulator \cite{Sto18}, that captures several features of real traffic patterns such as the vehicle dimensions, vehicle types, vehicle velocity, as well as realistic car-following and car-overtaking models.

Figure~\ref{fig:Poisson_confirmation} shows the distribution for the inter-vehicle distances obtained in the simulator. The simulated traffic leads indeed to a configuration where the inter-vehicular distance is exponentially distributed, characterizing a PPP. The cumulative distribution function of the uniform distribution is also shown for comparison. This property holds for $\lambda_v = 14$ vehicles/km/lane, but can be generalized for any $\lambda_v$ small enough to remain in a free-flow regime, as well as any other number of lanes in the highway. Note that the results shown in Figure~\ref{fig:Poisson_confirmation} correspond to the inter-arrival distances for the projection of cars in the three lanes onto the given axis, hence although vehicles cannot be closer than their dimension permits on a given lane in the simulator, the projection of the vehicles' centers on the three lanes can be arbitrarily close.

\begin{figure}[h]
\vspace{0.0cm}
\centering
\includegraphics[trim={1.5cm 0 1.5cm 0},clip, width=0.9\columnwidth]{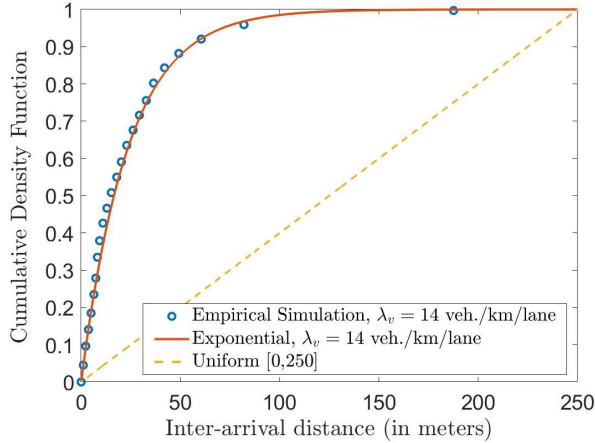}
\caption{Comparison of the simulated vehicle inter-arrival distance CDF with exponential and uniform random variables, on a collapsed 3 lanes highway system ($\eta = 3$).} \label{fig:Poisson_confirmation}
\end{figure}

Therefore, we expect that the observations and conclusions drawn from Figures~\ref{fig:results1}-\ref{fig:results_ge2} in the single lane scenario to apply in the multilane configuration as well. Moreover, our analysis in Sections~\ref{sec:multilane} and \ref{sec:multilane_perf_eval} predicts improved performances compared to the single lane case. For instance, we expect a higher probability of connectivity, better redundancy, or improved per-user shared rate for instance, due to the fact that clusters can be larger in size and that blocking vehicles have a less severe impact on the others.

\subsection{Insight on Alternative Distributions}
Although the PPP assumption will be a good fit in certain regimes, it will still fail for others that may arise in the future, e.g., where cars may intentionally form platoons to increase highway throughput. To better understand how such patterns might affect connectivity, in this section we ask the question ``What is the best possible configuration of cars, i.e., resulting in the best connectivity metrics?''. We shall focus on two performance metrics: coverage $\pi_v$ and mean rate per user. 
Two regimes can be distinguished. The first one corresponds to situations where $\lambda_v \geq 1/d$, i.e. where the vehicle density is large enough so that vehicles can be separated by $1/d$ meters. In such a scenario, vehicles would form a single infinite cluster leading to $\pi_v = 1$ and maximum mean rate per user since all the RSUs are in use. The other regime of interest is where $\lambda_v < 1/d$. Consider first a configuration where all the clusters in the network are of same size. Then spacing the vehicles by $d$ within the cluster would ensure maximal cluster length, and hence maximal $\pi_v$ and $\mathbb{E}[R_v]$ as this would maximize the ``space covered" by clusters and thus the RSU busy time. Similarly, spacing vehicles in adjacent clusters by $2d$ would also maximize $\mathbb{E}[R_v]$, without affecting the coverage.
Following these two rules, we derive expressions for $\pi_v$ and $u$, the average RSU utilization capturing the same information as $\mathbb{E}[R_v]$. For a fixed cluster size $n$:

\begin{equation}
\pi_v(n) = \min [(n+1)\cdot d \cdot \lambda_r, 1] 
\end{equation}
\begin{equation}
u(n) = \min [\frac{n+1}{n}\cdot d \cdot \lambda_v, 1]
\end{equation}

Clearly, as $n$ increases, $\pi_v(n)$ increases while $u(n)$ decreases. We exhibit that trend through a tradeoff curve between coverage and throughput as a function of $n$ in Figure \ref{fig:tradeoff_curves}:

\begin{figure}[h]
\vspace{0.0cm}
\centering
\includegraphics[width=1\columnwidth]{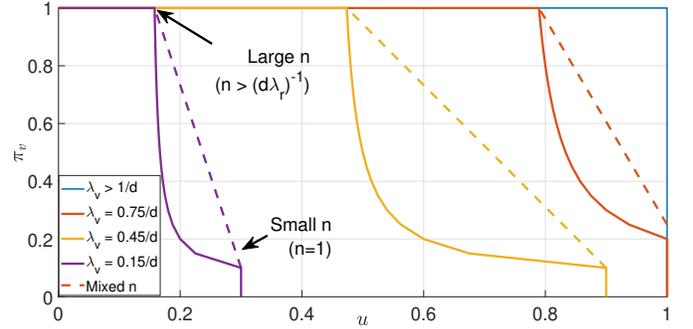}
\caption{Tradeoff curve between connectivity $\pi_v$ and RSU utilization $u$, for different $\lambda_v$ (in vehicles/km), and the achievable performance by mixing cluster sizes.} \label{fig:tradeoff_curves}
\end{figure}

Figure~\ref{fig:tradeoff_curves} exhibits the tradeoff between connectivity and throughput. In a low density regime, vehicles form longer clusters but cover less area as the cluster size $n$ increases, improving the connectivity but reducing the average RSU utilization, and hence the mean rate per user. We note that when the vehicle density $\lambda_v$ is large enough, the tradeoff does not occur as vehicles can get full connectivity and maximum mean rate per user.
In scenarios where cluster size mixing is allowed, cluster can see an even better performance, represented by a straight line between any two points on the tradeoff curves. We note that the best mixing possible is combinations of clusters of size 1, i.e. isolated vehicles, and clusters of size $n = \mathlarger{\mathlarger{\mathlarger{\lfloor}}} \frac{1}{d\lambda_r}+1 \mathlarger{\mathlarger{\mathlarger{\rfloor}}}$. The tradeoff curves associated with such mixings are drawn as dashed lines on Figure~\ref{fig:tradeoff_curves}. Intuitively, clusters of size 1 help to maximize the total area covered by the clusters, while the largest clusters increase the connectivity probability of a typical vehicle. Different combinations of those two cluster sizes can be constituted to reach any specific connectivity or throughput target.

%% file: conclusion.tex
In this work we have analyzed the performance of a multi-homed V2V+V2I architecture. Our main conclusion is that V2V relay clusters along with RSU multi-homing improves significantly the typical vehicle coverage probability and reliability, while reducing the variability of the shared rate per user when compared to a traditional V2I architecture. These properties position this architecture as a critical enabler for Internet connectivity services in future vehicular networks. We also conclude that the V2V technology penetration level is critical in the system performance given that many legacy vehicles will obstruct the LoS and prevent some vehicles to communicate. These difficulties may be mitigated if dedicated lanes are used by new vehicles that are V2V+V2I capable, particularly at low penetration levels.
Moreover, we proposed a new mechanism to bound the performance of multi-lane highways by equivalent single lane highways, and our simulation results highlight a robustness of performance to heterogeneous vehicle distributions across lanes. Finally, we described how the results presented throughout the paper would change if one could control the relative positions of the vehicles on the road, e.g., when autonomous vehicles form platoons, and how the connectivity-throughput tradeoff can be formally characterized in such scenarios.

%% file: appendix.tex
\begin{center}
\small
DEFINITIONS AND PROOFS
\end{center}

\begin{definition}\textbf{(Stochastic Dominance)}\label{def:dominance}
As in \cite{MuS02}, we define stochastic dominance as 
$$X\le^{st}Y \implies P(X>x)\le P(Y>x), ~ \forall x$$ 
and increasing convex dominance $$X\le^{icx}Y \implies \mathbb{E}[f(X)]\le \mathbb{E}[f(Y)], ~ \forall f\in \mathcal{F},$$
where $\mathcal{F}$ is the set of increasingly convex functions for which the expected value is defined.
Note further that if $X\le^{icx}Y$ and $\mathbb{E}[X]=\mathbb{E}[Y]$ then $X\le^{cx}Y$ for convex functions, e.g., 
$\mbox{Var}(X)\le \mbox{Var}(Y).$
\end{definition}

{\bf{Proof. \textit{(Lemma \ref{lem:mdi}})}} 
We denote as $\varphi$ the probability of not having any V2V capable vehicle in the communication range $d$ ahead of a typical vehicle.
From the Poisson assumption, the distance between vehicles follows an exponential distribution denoted by a random variable $E\sim\text{exp}(\lambda_v)$. The probability of having one or more vehicles within the communication range $d$ of a participant is then: $F_{E}(d) = 1-e^{-\lambda_v d}$.
Since the market penetration is considered independent of the interarrival time, the probability of the next car being a V2V+V2I capable vehicle within the communication range $d$ is given by $\gamma(1-e^{-\lambda_v d})$, thus $\varphi=1-\gamma (1-e^{-\lambda_v d})$.

Now since the number of users in a cluster is determined by the number of the successive V2V capable vehicles in range of each other, $p_{N}(n)=\varphi\left(1-\varphi\right)^{n-1},$
i.e., $N$ is a geometric random variable with parameter $\varphi$, and mean $\mathbb{E}[N]={1}/{\varphi}$.

From the analysis in \cite{WBM07}, it is well known that the average cluster communication range is $ \mathbb{E}[L]={\lambda_v}^{-1}\cdot\left(e^{\lambda_v d}-\lambda_v d-1\right)$ (defined as distance between the first and the last vehicle plus 2 times the communication range). However, the {density function} of the length has been only evaluated via simulations \cite{Rei14}.
The length of the cluster, given that there are $N$ vehicles, corresponds to $L=2d+\sum_{i=1}^{N-1} T_i,$ where $T_i$ denotes the inter-spacing of \textit{V2V} capable vehicles in the same cluster. Note that the distribution of $T_i$ is that of an exponential conditioned on being smaller than $d$, thus
$$f_{T_i}(l)=\frac{\lambda e^{-\lambda_v l}}{1-e^{-\lambda_v d}}, \quad 0<l\le d.$$
The moment generating function for $T_i$ is thus
$$M_{T_i}(s)=\int_{0}^{d} e^{sl} \frac{\lambda e^{-\lambda_v l}}{1-e^{-\lambda_v d}}dl= \frac{\lambda_v e^{d (s-\lambda_v)}-\lambda_v}{(s-\lambda_v)(1-e^{-\lambda_v d})}$$
and consequently the conditional moment generating function of the length of a cluster, given its number of vehicles, denoted as $M_{L\mid N=n}(s)$ is given by:
\begin{align*}
M_{L\mid N=n}(s)&=e^{2 s d}\prod_{i=1}^{n-1}M_{T_i}(s) \\
& =e^{2 s d}\left[\frac{\lambda_v e^{d (s-\lambda_v)}-\lambda_v}{(s-\lambda_v)(1-e^{-\lambda_v d})}\right]^{n-1}.
\label{eq:LgivenN}
\end{align*}

Given the conditional distribution, we can compute the moment generating function of $L$ via:
\vspace{-0.2cm}
\begin{align*}
M_{L}(s)&=\sum\limits_{n=1}^{\infty} M_{L\mid N=n}(s) p_{N}(n)
&=\frac{e^{ 2sd} \varphi}{1-M_T(s)+\varphi M_T(s)}.
\end{align*}

For the case of full market penetration, this simplifies to:
\begin{align*}
M_{L}(s)
=\frac{e^{d (2s-\lambda_v)} \left(s-\lambda_v\right)}{s-\lambda_v e^{d (s-\lambda_v)}}.
\end{align*}

The distributions ${f}_{L}(l)$ and ${f}_{L\mid N}(l)$ can be then obtained by the inverse Laplace transform of $M_{L}(-s)$ and $M_{L\mid N=n}(-s)$, respectively. The conditional CDF of the number of RSUs $M$ serving a cluster of length $L$  is given by
\begin{align*}
F_{M\mid L}(m\mid l)
&=\begin{cases}
1 & \mbox{if } m\lambda_r^{-1}<l,\\
1-\frac{l}{m\cdot \lambda_r^{-1}}   & \mbox{if } (m-1) \lambda_r^{-1}< l \le m\cdot \lambda_r^{-1},\\
0 &  \mbox{otherwise} 
\end{cases}
\end{align*}
since the cluster process is stationary on the line and independent of the RSU locations. Given $N$ it is direct to see, applying the chain rule that:
\begin{align*}
F^c_{M \mid N} \left(m\mid n\right)=\int\limits_{0}^{\infty}  F^c_{M\mid L}(m\mid l) f_{L\mid N}(l\mid n)dl.
\end{align*}
where $F^c(\cdot)$ stands for the complementary CDF.
Substituting the distributions obtained above, the simplified expression is given by \eqref{eq:M_given_N}.
\qed\\

{\bf{Proof. \textit{(Lemma \ref{lem:coverage}})}} 
In order to derive an expression for the probability of coverage of a typical vehicle $\pi_v;$ we will first relate the number of vehicles as seen in a typical cluster $N$ to that seen by a typical vehicle in its cluster $N_v:$
$$p_{N_v}(n)=\frac{n P(N=n)}{\mathbb{E}[N]},$$
represents the probability for a typical vehicle to be in a cluster of size $n$, where the $\frac{n}{\mathbb{E}[N]}$ biases the distribution of $N$ as a typical vehicle is more likely to belong to larger clusters. Therefore
{\begin{align*}
\pi_v&=\sum\limits_{n=1}^{\infty} p_{N_v}(n)  F^c_{M\mid N}\left(0\mid n\right)\\&=\varphi^2 \sum\limits_{n=1}^{\infty} n (1-\varphi)^{n-1} F^c_{M\mid N}\left(0\mid n\right).
\end{align*}}
In the case with networks with only \textit{V2I} links enabled, the probability that a typical vehicle is connected corresponds to the probability that the vehicle lands in the fraction of the road covered by RSUs given by:
$\pi_v^{*}=\frac{2d}{\lambda_r^{-1}}.$\qed\\

{\bf{Proof. \textit{(Theorem \ref{thm:throughput}})}} 
Under our sharing model, vehicles in a typical cluster with $(N,M)$ users and RSU connections will see a shared rate no larger than $\frac{\rho^{RSU}M}{N}.$ This is exact if the cluster does not share any of the RSUs with another cluster; otherwise this is an upper bound. Note that an RSU can be shared by two clusters, each approaching from one side, and both not close enough to form one larger cluster.
The mean rate seen by a typical vehicle $R_v$ is then bounded by:
\begin{align*}
\mathbb{E}[R_v]&\le{\mathbb{E}\left[\frac{N}{\mathbb{E}[N]}\frac{\rho^{RSU}}{N} M\right]}
=\frac{\rho^{RSU}\cdot \mathbb{E}[M] }{\mathbb{E}[N]}
\end{align*}
where once again we have moved from the typical cluster shared rate to the typical vehicle shared rate by weighting by $N/\mathbb{E}[N].$

In the case of the V2I network, the rate seen by a typical vehicle is given by:
\begin{align*}
\mathbb{E}[R_v^{*}]&=\mathbb{E}\left[\frac{\rho^{RSU}}{N^*+1}\mid I^*_v \right] \pi_v^*=\rho^{RSU} \frac{2d}{\lambda_r^{-1}} \mathbb{E}\left[\frac{1}{N^*+1}\right]
\end{align*}
where $I^*_v$ denotes the event of probability $\pi_v^*$ that a vehicle is connected, and $N^*$ denotes the number of (other) vehicles that a typical connected vehicle would see sharing its RSU. Note that the distribution of $N^*$, i.e., the reduced Palm distribution of the Poisson, is equal to its original distribution (Poisson $(2d\gamma\lambda_v)$), given the Slivnyak's Theorem \cite{BacV1}.
Therefore $\mathbb{E}\left[\frac{1}{N^*+1}\right]=\sum\limits_{n=0}^{\infty} \frac{P(N^*=n)}{n+1}
=\frac{1-e^{-2\gamma\lambda_v d}}{2\gamma\lambda_v d}$
and $\mathbb{E}[R_v^{*}]=\rho^{RSU} \cdot \frac{1-e^{-2\gamma\lambda_v d}}{\lambda_v \gamma\lambda_r^{-1}}$.

Finally, by coupling the vehicle locations for the V2V+V2I network and V2I network without relaying it is easy to observe that the number of busy RSUs is the same, so the mean rate seen by a typical vehicle in this two settings is the same, i.e., $\mathbb{E}[R_v]=\mathbb{E}[R_v^{*}].$\qed\\

{\bf{Proof. \textit{(Theorem \ref{thm:cdf_rate}})}} 
Paralleling Theorem \ref{thm:throughput}, an upper bound on the complementary CDF of the shared rate a typical vehicle sees in the V2V+V2I architecture, for $r>0$ is given by:
{\begin{align*}
F^c_{R_v}(r)&=P(R_v>r)\le \mathbb{E}\left[\frac{N}{\mathbb{E}[N]} \mathbb{E}\left[\mathbf{1}\left(\frac{M \rho^{\text{RSU}}}{N}\ge r\right)\mid N\right]\right]\\
&=\sum\limits_{n=1}^{\infty} \frac{n}{\mathbb{E}[N]} \cdot p_N(n)\cdot  F^c_{M\mid N}\left(\Bigl\lceil\frac{r n}{\rho^{\text{RSU}}}\Bigr\rceil\mid N=n\right)\\
&= \varphi^2\sum\limits_{n=1}^{\infty}  n \cdot(1-\varphi)^{n-1}\cdot  F^c_{M\mid N}\left(\Bigl\lceil\frac{r n}{\rho^{\text{RSU}}}\Bigr\rceil\mid N=n\right)  ,
\end{align*}
}
where $F^c(\cdot)$ stands for the complementary CDF.
Therefore, Equation \eqref{eq:cdfR} holds
and  $P(R_v=0)=1-\pi_v.$
Similarly the complementary CDF for the shared rate for a typical vehicle in the V2I network, for $r>0$ is
{ \begin{align}
P(R_v^{*}>r)&= P\left(\frac{\rho^{RSU}}{N^*+1}>r\right) \frac{d}{\lambda_r^{-1}}\nonumber\\&=P\left(\frac{\rho^{RSU}-r}{r}>{N^*}\right) \frac{d}{\lambda_r^{-1}}\nonumber\\
&=\frac{2 d}{\lambda_r^{-1}}\sum_{i=0}^{\lfloor \frac{\rho^{RSU}-r}{r} \rfloor}\left(\frac{(2\gamma\lambda_v d)^i}{i!}e^{-2 \gamma\lambda_v  d} \right)\nonumber\nonumber\\&
=\frac{2 d}{\lambda_r^{-1}}\cdot Q\left( \frac{\rho^{RSU}-r}{r},2\gamma\lambda_v d \right)\nonumber
\end{align}}
where $Q(.)$ is the regularized gamma function. 
Consequently, Equation \eqref{eq:cdfR_noV2V} holds
and  $P(R_v^*=0)=1-\pi_v^*.$
In order to prove the increasing convex dominance relation, we can use a coupling argument. We generate a single lane highway instance. It is clear that, for this instance, the number of vehicles and the total rate out of the network is the same, but the clusters are bigger in the V2V+V2I system (since the V2I only system only have clusters of one vehicle). It is proven in \cite{RaB07} that a max-min fairness allocation achieves the lexicographically minimum vector, i.e., for a max-min share rate allocation $\hat{R}$ and any other shared rate allocation ${R}$ then $\hat{R}$ \textit{is majorized} by ${R}$ \cite{HLP29} and further implies $\hat{R}\le^{icx}{R}.$ The proof is then completed by noticing that the max-min shared rate allocation of the V2I system $R^{*}$ is a feasible rate allocation in the V2V+V2I system, so ${R}\le^{icx}{R^*}.$
\qed\\

{\bf{Proof. \textit{(Theorem \ref{thm:multilane2}})}}
The proof relies on constructing a coupling between a random process $\xi^{\mathcal{M}}$ denoting vehicle locations on a multilane highway  $\mathcal{M}$ and an auxiliary process $\xi^{\mathcal{S}}$ denoting their locations on a single lane highway  $\mathcal{S}.$ 

Let $(T_i,K_i)_{i\in\mathbb{N}}$ denote the sequence of locations of V2V+V2I capable vehicles on $\mathcal{M}$, where $T_i$ denotes the location of the $i^{\text{th}}$ vehicle and $K_i$ its associated lane. We define the aggregated V2V+V2I capable vehicle intensity on highway $\mathcal{M}$ as $\lambda^{\text{V2V}}$, and their intensity on lane $k$ as $\lambda^{\text{V2V}}_k$. Note that under our Poisson assumption $(T_i)_{i\in\mathbb{N}}$ is a $\text{PPP}(\lambda^{\text{V2V}})$ and $(K_i)_{i\in\mathbb{N}}$ is distributed as $$p_{K_i}(k)=\left(\frac{\lambda^{\text{V2V}}_k}{\lambda^{\text{V2V}}}: k=1,2,\ldots,\eta\right),   \forall i\in\mathbb{N}$$ \noindent since aggregation of independent PPPs is also a PPP.

The first step of the coupled single-lane highway $\mathcal{S}$ construction consists in including V2V+V2I capable vehicles at locations $(T_i)_{i\in\mathbb{N}}$ in the auxiliary process $\xi^{\mathcal{S}}.$ 

Let us now consider the blocking vehicles in the multilane highway $\mathcal{M}$. These vehicles also correspond to PPPs of intensity $\lambda^b_k$ on each lane $k$ and independent of $(T_i,K_i)_{i\in\mathbb{N}}.$ For a given realization  $(t_i,k_i)_{i\in\mathbb{N}}$ let $B^{k_i}(t_i,t_{i+1}]$ denote a set of locations of blocking vehicles on lane $k$ in the time interval $(t_i,t_{i+1}]$ in the multilane highway. Note that $B^{k_i}(t_i,t_{i+1}]$ for any $i$ are mutually independent by the definition of $\text{PPP}.$

We shall let $B(t_i,t_{i+1}]$ denote blocking vehicles' locations that the process $\mathcal{M}$ will share with $\mathcal{S}$. Specifically, according to the blocking model in Definition 
\ref{block-model} we let
$$\hspace{-0.6cm}B(t_i,t_{i+1}]=\begin{cases}
\bigcup\limits_{j=\mathrm{k}_i+1}^{\mathrm{k}_{i+1}-1} B^{j}(t_i,t_{i+1}] & \text{if}~ \mathrm{k}_{i+1}>j>\mathrm{k}_{i}\\~\vspace{-0.35cm}\\
\quad B^{j}(t_i,t_{i+1}] & \text{if}~ \mathrm{k}_{i}=\mathrm{k}_{i+1}=j\\~\vspace{-0.35cm}\\
\bigcup\limits_{j=\mathrm{k}_{i+1}+1}^{\mathrm{k}_i-1}  B^{j}(t_i,t_{i+1}]  & \text{if}~ \mathrm{k}_i>j>\mathrm{k}_{i+1}\\~\vspace{-0.35cm}\\
~\qquad ~ \emptyset & \text{otherwise}.
\end{cases}.$$
Note that in each interval $(t_i,t_{i+1}]$, $B(t_i,t_{i+1}]$ are Poisson process independent but with different intensities depending on $k_i$ and $k_{i+1}.$ 

Note also that given our blocking model $B(\cdot,\cdot]$ includes all vehicles that may block connectivity of V2V+V2I capable cars in $\mathcal{M}$. Figure~\ref{fig:example_coupling} shows examples of configurations and their associated $B(.,.]$. Finally, for each $i\in\mathrm{N}$ we define 

\begin{figure*}[h]
\centering
\includegraphics[width=1.6\columnwidth]{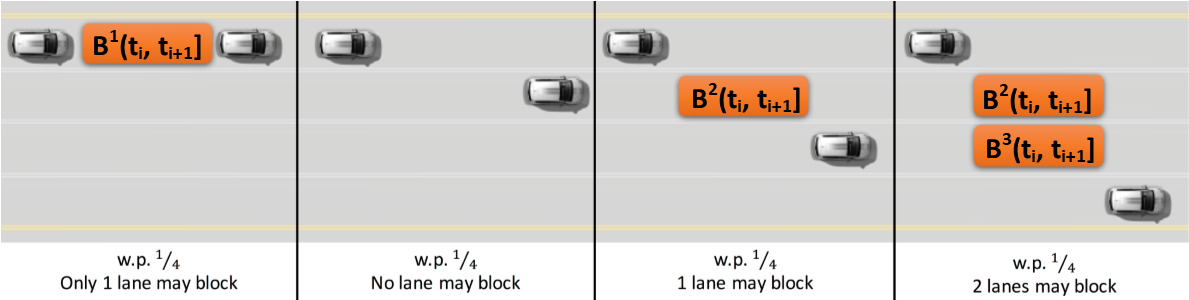}
\caption{Examples of configurations and their associated $B^i(.,.]$ sets, for $\eta = 4$ and $k_i \leq k_{i+1}$.}
\label{fig:example_coupling}
\end{figure*}

$$B^{\mathcal{S}}(t_i,t_{i+1}]=B(t_i,t_{i+1}]\cup A(t_i,t_{i+1}]$$
where $A(t_i,t_{i+1}]$ 
is an independent PPP on $(t_i,t_{i+1}]$ with intensity needed to ensure that the overall intensity is equalized in all intervals; ensuring that $B^{\mathcal{S}}(t_i,t_{i+1}]$ is a PPP with intensity $\lambda^b_{\text{eff}}.$ 
 We shall introduce $B^{\mathcal{S}}(t_i,t_{i+1}]_{i\in\mathbb{N}}$ in each of the intervals in the process $\xi^{\mathcal{S}}.$

At this point, it is worth noting that given our construction, 
\begin{equation}
\text{LoS interrupted in } \xi^{\mathcal{M}} ~ ~ \begin{aligned}
  {\implies} \\
  {\centernot\impliedby} 
\end{aligned} ~ ~ \text{LoS interrupted in } \xi^{\mathcal{S}}.
\label{eq:cond}
\end{equation}
and the distributions of $\xi^{\mathcal{M}}\sim{\mathcal{M}}=\mathcal{H}(\eta, \boldsymbol{{\lambda^{\text{V2V}}}}, ~\boldsymbol{{\lambda^{b}}})$ and $\xi^{\mathcal{S}}\sim{\mathcal{S}}=\mathcal{H}(1, ~\gamma\lambda,\lambda^b_{\text{eff}})$
where $${{\lambda}}=\lambda^{\text{V2V}}+\lambda^{b},\  \gamma=\frac{\lambda^{\text{V2V}}}{\lambda} \  
\text{and} \ 
\lambda^b_{\text{eff}} = \max(\lambda^{b}_0, \lambda^{b}_k, \sum^{\eta-1}_{i=2}\lambda^{b}_{i}).$$
This implies the following fact.
\begin{fact}\label{fact:icx} Based on the aforementioned coupling one can show that
$N_v^\mathcal{M}\ge^{st}N_v^\mathcal{S},\quad L_v^\mathcal{M}\ge^{st}L_v^\mathcal{S}\quad \text{and}\quad M_v^\mathcal{M}\ge^{st}M_v^\mathcal{S},$
and,
$\pi_v^\mathcal{M}\ge \pi_v^\mathcal{S},\quad R_v^\mathcal{M}\le^{icx}R_v^\mathcal{S}.$\end{fact}
\begin{proof}
Note that by ergodicity of the cluster process, $P(N_v^\mathcal{M}>n)$ and $P(N_v^\mathcal{S}>n)$ correspond to:
$$P(N_v^\mathcal{M}>n)=\lim_{c\rightarrow\infty} \frac{1}{\sum\limits_{i=1}^{c} N_{i}^\mathcal{M} }\sum\limits_{i=1}^{c} N_{i}^\mathcal{M} \cdot \mathrm{\mathbf{1}}(N_{i}^\mathcal{M}>n)$$
$$P(N_v^\mathcal{S}>n)=\lim_{c\rightarrow\infty} \frac{1}{\sum\limits_{i=1}^{c} N_{i}^\mathcal{M} }\sum\limits_{i=1}^{c} \sum\limits_{j=1}^{Y_i} N_{i,j}^\mathcal{S} \cdot\mathrm{\mathbf{1}}(N_{i,j}^\mathcal{S}>n),$$
where $N_{i}^{\mathcal{M}}$ is the number of vehicles in the $i^{\text{th}}$ cluster in the multilane and $Y_i$ is the number of subclusters in the single lane originated from the $i^{\text{th}}$ cluster in the multilane. $N^{\mathcal{S}}_{i,j}$ denotes the number of vehicles in the $j$th subcluster in the single lane process. 

By noting that the clusters in $\mathcal{S}$ are created by splitting the clusters of $\mathcal{M}$, we can see that
$$\mathrm{\mathbf{1}}(N_{i}^\mathcal{M}>n)\ge \mathrm{\mathbf{1}}(N_{i,j}^\mathcal{S}>n),\quad \forall i,j.$$
and given the fact that $N_{i}^\mathcal{M}=\sum\limits_{j=1}^{Y_i} N_{i,j}^\mathcal{S},$ we have that
\begin{align*}
P(N_v^\mathcal{M}>n)&=\lim_{c\rightarrow\infty} \frac{1}{\sum\limits_{i=1}^{c} N_{i}^\mathcal{M} }\sum\limits_{i=1}^{c} \sum\limits_{j=1}^{Y_i} N_{i,j}^\mathcal{S} \cdot \mathrm{\mathbf{1}}(N_{i}^\mathcal{M}>n)\\
&\ge \lim_{c\rightarrow\infty} \frac{1}{\sum\limits_{i=1}^{c} N_i^\mathcal{M} }\sum\limits_{i=1}^{c} \sum\limits_{j=1}^{Y_i} N_{i,j}^\mathcal{S} \cdot \mathrm{\mathbf{1}}(N_{i,j}^\mathcal{S}>n) \\
&=P(N_v^\mathcal{S}>n)
\end{align*}
and therefore $N_v^\mathcal{M}\ge^{st}N_v^\mathcal{S}.$
Similarly, $$\quad L_v^\mathcal{M}\ge^{st}L_v^\mathcal{S}\quad \text{and}\quad M_v^\mathcal{M}\ge^{st}M_v^\mathcal{S}$$
by noting that
$
\mathrm{\mathbf{1}}(L_{v,i}^\mathcal{M}>l)		\ge\mathrm{\mathbf{1}}(L_{v,i,j}^\mathcal{S}>l) $
and
$
 \mathrm{\mathbf{1}}(M_{v,i}^\mathcal{M}>m)	\ge\mathrm{\mathbf{1}}(M_{v,i,j}^\mathcal{S}>m)
$
are direct implications of Eq. \eqref{eq:cond}.
Additionally it also has the implication that, within a cluster, if we denote as $\pi_{v,i}$ the probability that a typical vehicle in cluster $i$th is connected then
$\pi_{v,i}^\mathcal{M} \ge \pi_{v,i,j}^\mathcal{S}.$

%
Noting that $N_{i}^\mathcal{M}=\sum\limits_{j=1}^{Y_i} N_{i,j}^\mathcal{S}$ and observing that the expected shared rate per vehicle is equal in both systems we can directly infer the $R_v^\mathcal{M}\le^{icx}R_v^\mathcal{S}$. It is proven in \cite{RaB07} that a max-min fairness allocation achieves the lexicographically minimum vector, i.e., for a max-min share rate allocation $\hat{R}$ and any other shared rate allocation ${R}$ then $\hat{R}$ \textit{is majorized} by ${R}$ \cite{HLP29} and further implies $\hat{R}\le^{icx}{R}.$ 
The max-min shared rate allocation of the single lane system $R^{\mathcal{S}}$ is always a feasible rate allocation in the multilane system; since the single lane system has the same number of vehicles and the same mean rate, but the ability for the vehicles to reach the RSUs is reduced and we have that $R_v^\mathcal{M}\le^{icx}R_v^\mathcal{S}.$

%
\end{proof}

\vspace{-0.4cm}

%% file: bare_jrnl.bbl
\begin{thebibliography}{10}
\providecommand{\url}[1]{#1}
\csname url@samestyle\endcsname
\providecommand{\newblock}{\relax}
\providecommand{\bibinfo}[2]{#2}
\providecommand{\BIBentrySTDinterwordspacing}{\spaceskip=0pt\relax}
\providecommand{\BIBentryALTinterwordstretchfactor}{4}
\providecommand{\BIBentryALTinterwordspacing}{\spaceskip=\fontdimen2\font plus
\BIBentryALTinterwordstretchfactor\fontdimen3\font minus
  \fontdimen4\font\relax}
\providecommand{\BIBforeignlanguage}[2]{{%
\expandafter\ifx\csname l@#1\endcsname\relax
\typeout{** WARNING: IEEEtran.bst: No hyphenation pattern has been}%
\typeout{** loaded for the language `#1'. Using the pattern for}%
\typeout{** the default language instead.}%
\else
\language=\csname l@#1\endcsname
\fi
#2}}
\providecommand{\BIBdecl}{\relax}
\BIBdecl

\bibitem{ZZC15}
K.~Zheng, Q.~Zheng, P.~Chatzimisios, W.~Xiang, and Y.~Zhou, ``{Heterogeneous
  Vehicular Networking: A Survey on Architecture, Challenges, and Solutions},''
  \emph{IEEE Communications Surveys Tutorials}, 2015.

\bibitem{Ken11}
J.~B. Kenney, ``{Dedicated Short-Range Communications (DSRC) Standards in the
  United States},'' \emph{Proceedings of the IEEE}, July 2011.

\bibitem{LiW07}
F.~Li and Y.~Wang, ``{Routing in vehicular ad hoc networks: A survey},''
  \emph{IEEE Vehicular Technology Magazine}, June 2007.

\bibitem{BhK14}
S.~K. Bhoi and P.~M. Khilar, ``{Vehicular communication: a survey},'' \emph{IET
  Networks}, September 2014.

\bibitem{CCC14}
A.~M. Cailean, B.~Cagneau, L.~Chassagne, V.~Popa, and M.~Dimian, ``{A survey on
  the usage of DSRC and VLC in communication-based vehicle safety
  applications},'' in \emph{2014 IEEE 21st Symposium on Communications and
  Vehicular Technology in the Benelux (SCVT)}, November 2014.

\bibitem{CVG16}
J.~Choi, V.~Va, N.~Gonzalez-Prelcic, R.~Daniels, C.~R. Bhat, and R.~W. Heath,
  ``{Millimeter-Wave Vehicular Communication to Support Massive Automotive
  Sensing},'' \emph{IEEE Communications Magazine}, Dec. 2016.

\bibitem{CHS16}
S.~{Chen}, J.~{Hu}, Y.~{Shi}, and L.~{Zhao}, ``{LTE-V: A TD-LTE-Based V2X
  Solution for Future Vehicular Network},'' \emph{IEEE Internet of Things
  Journal}, Dec 2016.

\bibitem{NiH10}
D.~Niyato and E.~Hossain, ``{A Unified Framework for Optimal Wireless Access
  for Data Streaming Over Vehicle-to-Roadside Communications},'' \emph{IEEE
  Transactions on Vehicular Technology}, July 2010.

\bibitem{LMN13}
Y.~Liu, J.~Ma, J.~Niu, Y.~Zhang, and W.~Wang, ``{Roadside units deployment for
  content downloading in vehicular networks},'' in \emph{2013 IEEE
  International Conference on Communications (ICC)}, June 2013.

\bibitem{YMY10}
Y.~Yang, Z.~Mi, J.~Y. Yang, and G.~Liu, ``{A Model Based Connectivity
  Improvement Strategy for Vehicular Ad hoc Networks},'' in \emph{2010 IEEE
  72nd Vehicular Technology Conference - Fall}, September 2010.

\bibitem{CaD17}
A.~Cailean and M.~Dimian, ``{Current Challenges for Visible Light
  Communications Usage in Vehicle Applications: A Survey},'' \emph{IEEE
  Communications Surveys Tutorials}, May 2017.

\bibitem{AbZ11}
A.~Abdrabou and W.~Zhuang, ``{Probabilistic Delay Control and Road Side Unit
  Placement for Vehicular Ad Hoc Networks with Disrupted Connectivity},''
  \emph{IEEE Journal on Selected Areas in Communications}, Jan. 2011.

\bibitem{RSN14}
A.~B. Reis, S.~Sargento, F.~Neves, and O.~K. Tonguz, ``{Deploying Roadside
  Units in Sparse Vehicular Networks: What Really Works and What Does Not},''
  \emph{IEEE Transactions on Vehicular Technology}, vol.~63, July 2014.

\bibitem{AKA17}
R.~Atallah, M.~Khabbaz, and C.~Assi, ``{Multihop V2I Communications: A
  Feasibility Study, Modeling, and Performance Analysis},'' \emph{IEEE
  Transactions on Vehicular Technology}, March 2017.

\bibitem{CML17}
J.~{Chen}, G.~{Mao}, C.~{Li}, A.~{Zafar}, and A.~Y. {Zomaya}, ``Throughput of
  infrastructure-based cooperative vehicular networks,'' \emph{IEEE
  Transactions on Intelligent Transportation Systems}, Nov 2017.

\bibitem{SLZ14}
C.~Shao, S.~Leng, Y.~Zhang, A.~Vinel, and M.~Jonsson, ``{Analysis of
  connectivity probability in platoon-based Vehicular Ad Hoc Networks},'' in
  \emph{2014 International Wireless Communications and Mobile Computing
  Conference (IWCMC)}, August 2014.

\bibitem{ZCG16}
J.~Zhao, Y.~Chen, and Y.~Gong, ``{Study of Connectivity Probability of
  Vehicle-to-Vehicle and Vehicle-to-Infrastructure Communication Systems},'' in
  \emph{2016 IEEE 83rd Vehicular Technology Conference (VTC Spring)}, May 2016.

\bibitem{PaW17}
B.~{Pan} and H.~{Wu}, ``Performance analysis of connectivity considering user
  behavior in v2v and v2i communication systems,'' in \emph{2017 IEEE 86th
  Vehicular Technology Conference (VTC-Fall)}, Sep. 2017.

\bibitem{KKS16}
S.~Kwon, Y.~Kim, and N.~B. Shroff, ``{Analysis of Connectivity and Capacity in
  1-D Vehicle-to-Vehicle Networks},'' \emph{IEEE Transactions on Wireless
  Communications}, December 2016.

\bibitem{WBM07}
N.~Wisitpongphan, F.~Bai, P.~Mudalige, V.~Sadekar, and O.~Tonguz, ``{Routing in
  Sparse Vehicular Ad Hoc Wireless Networks},'' \emph{IEEE Journal on Selected
  Areas in Communications}, October 2007.

\bibitem{SoT11}
S.~I. Sou and O.~K. Tonguz, ``Enhancing vanet connectivity through roadside
  units on highways,'' \emph{IEEE Transactions on Vehicular Technology}, Oct
  2011.

\bibitem{GSH11}
{M. Gramaglia and P. Serrano and J. A. Hernandez and M. Calderon and C. J.
  Bernardo}, ``{New insights from the analysis of free flow vehicular traffic
  in highways},'' in \emph{2011 IEEE International Symposium on a World of
  Wireless, Mobile and Multimedia Networks}, June 2011.

\bibitem{MaA09}
G.~{Mao} and B.~D.~O. {Anderson}, ``Graph theoretic models and tools for the
  analysis of dynamic wireless multihop networks,'' in \emph{2009 IEEE Wireless
  Communications and Networking Conference}, April 2009.

\bibitem{UGB15}
M.~{Uysal}, Z.~{Ghassemlooy}, A.~{Bekkali}, A.~{Kadri}, and H.~{Menouar},
  ``Visible light communication for vehicular networking: Performance study of
  a v2v system using a measured headlamp beam pattern model,'' \emph{IEEE
  Vehicular Technology Magazine}, Dec 2015.

\bibitem{KVB16}
Z.~Khan, A.~Vasilakos, B.~Barua, S.~Shahabuddin, and H.~Ahmadi, ``Cooperative
  content delivery exploiting multiple wireless interfaces: methods, new
  technological developments, open research issues and a case study,''
  \emph{Wireless networks}, 2016.

\bibitem{BGH87}
D.~P. Bertsekas, R.~G. Gallager, and P.~Humblet, \emph{Data networks}.\hskip
  1em plus 0.5em minus 0.4em\relax Prentice-hall Englewood Cliffs, NJ, 1987.

\bibitem{NLJ15}
Y.~Niu, Y.~Li, D.~Jin, L.~Su, and A.~Vasilakos, ``{A survey of millimeter wave
  communications (mmWave) for 5G: opportunities and challenges},''
  \emph{Wireless Networks}, vol.~21, 2015.

\bibitem{PFH15}
P.~H. Pathak, X.~Feng, P.~Hu, and P.~Mohapatra, ``Visible light communication,
  networking, and sensing: A survey, potential and challenges,'' \emph{IEEE
  Communications Surveys Tutorials}, 2015.

\bibitem{Sto18}
\BIBentryALTinterwordspacing
P.~Stone. {AIM4: Autonomous Intersection Management}. [Online]. Available:
  \url{http://www.cs.utexas.edu/~aim/, Accessed on 2018-10-01}
\BIBentrySTDinterwordspacing

\bibitem{MuS02}
A.~M{\"u}ller and D.~Stoyan, \emph{Comparison Methods for Stochastic Models and
  Risks}, ser. Wiley Series in Probability and Statistics.\hskip 1em plus 0.5em
  minus 0.4em\relax Wiley, 2002.

\bibitem{Rei14}
\BIBentryALTinterwordspacing
A.~B. Reis, ``{Numerical Computation of Cluster Length PDF}.'' [Online].
  Available: \url{http://users.ece.cmu.edu/~abragare/vehicular/}
\BIBentrySTDinterwordspacing

\bibitem{BacV1}
F.~Baccelli and B.~Blaszczyszyn, \emph{{Stochastic Geometry and Wireless
  Networks, Volume I - Theory}}, ser. Foundations and Trends in Networking Vol.
  3: No 3-4, pp 249-449.\hskip 1em plus 0.5em minus 0.4em\relax {NoW
  Publishers}, 2009, vol.~1.

\bibitem{RaB07}
B.~Radunovic and J.~Y.~L. Boudec, ``A unified framework for max-min and min-max
  fairness with applications,'' \emph{IEEE/ACM Transactions on Networking}, Oct
  2007.

\bibitem{HLP29}
G.~Hardy, J.~Littlewood, and G.~Polya, ``Some simple inequalities satisfied by
  convex functions,'' \emph{Messenger Math.}, 1929.

\end{thebibliography}
